\documentclass{siamltex}

\usepackage{amssymb,amsmath,epsfig,fleqn}
\usepackage{times,pifont,type1cm,eso-pic}

\usepackage{color,bm} 
\usepackage[dvips,ps2pdf]{hyperref}
\hypersetup{backref,hyperindex=true,pageanchor=true,colorlinks=true,citecolor=blue,urlcolor=blue}

\newtheorem{defin}[theorem]{Definition}
\newtheorem{cor}[theorem]{Corollary}
\newtheorem{remark}[theorem]{Remark}
\newtheorem{conj}[theorem]{Conjecture}
\newtheorem{concl}[theorem]{Conclusion}

\newcommand{\SH}{Schr\"odinger\,}
\newcommand{\GP}{Gross-Pitaevski\u i\,}
\newcommand{\ie}{{\it i.e.\ }}
\newcommand{\eg}{{\it e.g.\ }}
\newcommand{\sech}{{\rm sech}}

\newcommand{\hide}[1]{}

\newcommand{\fig}[4]{
\begin{figure}[ht!]
  \centering
  \includegraphics[width=#3\columnwidth]{./#1.eps} % subdirectory
  \caption{#4}
  \label{fig:#2}
\end{figure}
}

\renewcommand{\iff}{\Leftrightarrow}

\newcommand{\cD}{{\cal D}}

\newcommand{\pd}[2]{\frac{\partial #1}{\partial #2}}
\newcommand{\pdd}[2]{\frac{\partial^2 #1}{\partial #2^2}}
\newcommand{\dd}[2]{\frac{d #1}{d #2}}

\newcommand{\eps}{\varepsilon}

\newcommand{\R}{\mathbb{R}}

\newcommand{\ft}{\tilde{f}}

\newcommand{\kmax}{k_{\rm max}}
\newcommand{\gmax}{\Gamma_{\rm max}}

\newcommand{\pcr}{p_{\rm cr}}
\newcommand{\kcr}{k_{\rm cr}}

\newcommand{\Ocr}{\Omega_{\rm cr}}
\newcommand{\wcr}{w_{\rm cr}}

\newcommand{\vphib}{\boldsymbol{\varphi}}

\newcommand{\C}{\mathbb{C}}
\newcommand{\Mcr}{M_{\rm cr}}
\newcommand{\thcr}{\theta_{\rm cr}}
\newcommand{\bcr}{\beta_{\rm cr}}
\newcommand{\thsonic}{\theta_{\rm sonic}}
\newcommand{\pcut}{p_{\rm cutoff}}
\newcommand{\kcut}{k_{\rm cutoff}}

\newcommand{\avg}[1]{\langle #1 \rangle}
\newcommand{\inner}[2]{\avg{#1,#2}}
\newcommand{\fzero}{\mathbf{f}_0}
\newcommand{\fzeroad}{R \sigma_1 \mathbf{f}_0^*}
\newcommand{\varprod}[1]{ \frac{ \inner{#1}{\fzeroad} } { \inner{\fzero}{\fzeroad} } }

\definecolor{magenta}{rgb}{1,0,1}

\begin{document}

%\begin{frontmatter}

\title{Dark solitons, dispersive shock waves, and transverse
  instabilities}

% \author[NCSU]{Mark Hoefer}
% \author[UC]{Boaz Ilan}

% \address[NCSU]{Department of Mathematics,  Box 8205,
% North Carolina State University, NC 27695-8205}
% \address[UC]{School of Natural Sciences, 
% University of California, CA 95343}
\author{M. A. Hoefer\footnotemark[2] \and B. Ilan\footnotemark[3]}

\maketitle

\renewcommand{\thefootnote}{\fnsymbol{footnote}}

\footnotetext[2]{Department of Mathematics, North Carolina State
  University, Raleigh, NC 27695; mahoefer@ncsu.edu}
\footnotetext[3]{School of Natural Sciences, University of California,
  Merced, CA 95343; bilan@ucmerced.edu}

\renewcommand{\thefootnote}{\arabic{footnote}}

\begin{abstract}
  The nature of transverse instabilities to dark solitons and
  dispersive shock waves for the (2+1)-dimensional defocusing
  nonlinear \SH / \GP (NLS / GP) equation is considered.  Special
  attention is given to the small (shallow) amplitude regime, which
  limits to the Kadomtsev-Petviashvili (KP) equation.  We study
  analytically and numerically the eigenvalues of the linearized NLS /
  GP equation.  The dispersion relation for shallow solitons is
  obtained asymptotically beyond the KP limit.  This yields 1) the
  maximum growth rate and associated wavenumber of unstable
  perturbations; and 2) the separatrix between convective and absolute
  instabilities.  The latter result is used to study the transition
  between convective and absolute instabilities of oblique dispersive
  shock waves (DSWs).  Stationary and nonstationary oblique DSWs are
  constructed analytically and investigated numerically by direct
  simulations of the NLS / GP equation.  The instability properties of
  oblique DSWs are found to be directly related to those of the dark
  soliton.  It is found that stationary and nonstationary oblique DSWs
  have the same jump conditions in the shallow and hypersonic regimes.
  These results have application to controlling nonlinear waves in
  dispersive media.
\end{abstract}

% \begin{keywords}
%   Semiclassical regime \sep
%   traveling waves \sep
%   resonance \sep
%   Kadomtsev-Petviashvili equation \sep
%   Bose-Einstein condensates\sep
%   nonlinear optics
% \end{keywords}

%\end{frontmatter}

\pagestyle{myheadings} \thispagestyle{plain} \markboth{M. A. HOEFER
  AND B. ILAN}{DARK SOLITONS, DSWS, AND TRANSVERSE INSTABILITIES}

%\tableofcontents

%%%%%%%%%%%%%%%%%%%%%%%%%%%%%%%%%%%%%%%%%%%%%%%%%%%%%%

\section{Introduction}

The instability of one-dimensional structures to weak, long
wavelength, transverse perturbations plays an important role in
multi-dimensional nonlinear wave propagation.  Examples include
nonlinear optics \cite{kivshar_self-focusing_2000}, Bose-Einstein
condensates (BECs) \cite{frantzeskakis_dark_2010}, and water waves
\cite{bridges_transverse_2001,akers_dynamics_2010}.  Early theoretical
work on the transverse instability of solitons for the
Kadomtsev-Petviashvili (KP) equation
\cite{kadomtsev-petviashvili_1970,zakharov__1975} and the nonlinear
\SH (NLS) equation \cite{zakharov-rubenchik-74,kuznetsov-turitsyn-81}
focused on its existence and maximum growth rate, both properties of
the \textbf{real portion} of the spectrum of unstable modes.  Recent
numerical simulations of NLS \cite{el_oblique_2006} and vector NLS
\cite{gladush_wave_2009} supersonic flow past an obstacle in
two-dimensions reveal the excitation of apparently stable, oblique
spatial dark solitons for certain flow parameters. The resolution of
this inconsistency was explained in \cite{kamchatnov-pitaevskii-08},
where the instability was shown to be of \emph{convective} type so
that transverse perturbations are carried away by the flow parallel to
the soliton plane, effectively stabilizing the soliton near the
obstacle.  The characterization of convective versus absolute
instability requires knowledge of the spectrum for a range of
wavenumbers in the complex plane \cite{sturrock-58,briggs-64}.  For
NLS dark solitons, the criteria can be simplified and involve the
\textbf{imaginary} (stable) portion of the spectrum
\cite{kamchatnov-pitaevskii-08}.

One of the hallmarks of supersonic flow is the formation of shock
waves.  In classical, viscous fluids, shock dynamics can be well
understood mathematically in the context of a dissipative
regularization of conservation laws
(cf.~\cite{dafermos_hyperbolic_2009}).  There are, however, a number
of fluids with negligible dissipation whose dominant regularizing
mechanism is dispersion (see the review
\cite{hoefer_dispersive_2009}).  Most notably, superfluidic BECs and
optical waves in defocusing nonlinear media fall within this class of
dispersive fluids.  When a dispersive fluid flows at supersonic speed,
it can form a dispersive shock wave (DSW) that possesses an expanding,
oscillatory wavetrain with a large amplitude, soliton edge and small
amplitude sound wave edge.  DSWs appear as special, asymptotic
solutions of nonlinear dispersive equations and have been observed in
BEC
\cite{dutton_observation_2001,hoefer_dispersive_2006,meppelink_observation_2009}
and nonlinear optics
\cite{wan_dispersive_2007,conti_observation_2009}.
% , as undular bores in the atmosphere \cite{porter_modelling_2002},
% and as undular hydraulic jumps in shallow water
% \cite{chanson_current_2009}.
Their theory is much less developed than their classical (dissipative)
counterparts.  In particular, there has been limited study of DSW
stability.  Recent works numerically observe transverse instabilities
for NLS DSWs resulting from dark pulse propagation on a background in
two spatial dimensions \cite{armaroli_suppression_2009} and for
oblique DSWs in supersonic flow past a corner
\cite{hoefer_oblique_2009}.  In the former case, the transverse
instability was mitigated by introducing nonlocal nonlinearity while
in the latter case, the convective nature of the instability
effectively stabilizes the oblique DSW in certain parameter regimes.
In contrast, oblique shock waves in multidimensional, classical gas
dynamics are known to be linearly stable when the downstream flow is
supersonic
\cite{majda_stability_1983,li_analysis_2004,tkachev_courant-friedrichs_2009}
(see also the review article \cite{zumbrun_stability_2005} for more
general results).

The aim of this work is to clarify the role of absolute and convective
instabilities as they relate to spatial dark solitons and apply this
understanding to DSWs in multiple spatial dimensions.  Analytical and
computational challenges include: 
\begin{itemize}
\item The multi-dimensional nature of the flows.
\item The general criteria for absolute and convective instabilities
  requires detailed knowledge of the spectrum.
\item Long time integration and large spatial domains are required to
  properly resolve DSWs numerically.
\end{itemize}

% MAY WISH TO INCORPORATE
%
% Studies of DSWs date back to shallow water systems (see
% \eg~\cite{smyth_hydraulic_1988}) 
% and have been observed in multiple
% branches of physics including astrophysical plasma
% \cite{taylor_observation_1970}, ultra-cold atoms
% \cite{dutton_observation_2001,hoefer_dispersive_2006}, and nonlinear
% optics \cite{wan_dispersive_2007}.  The one-dimensional theory of
% Whitham averaging has been very successful in deriving the salient
% features of DSWs.  However, except for few recent studies, the entire
% body of DSW [theory and experiments] literature has been confined to
% \emph{one spatial dimension} .
% %This is due, in part,  to the difficulty in generalizing
% %Whitham averaging to multiple spatial dimensions owing to the
% %complexity of the resulting modulation equations~\cite{krichever_1988}.
% Recent experimental observations have demonstrated the formation of
% \emph{oblique DSWs}, which are inherently multi-dimensional
% structures.  This had led to a growing interest in the theory and
% computation of \emph{multi-dimensional
%   DSWs}\cite{el_oblique_2006,el_spatial_2006,
%   gladush_radiation_2007,kamchatnov-pitaevskii-08,hoefer_theory_2009}.
% In this study some of the salient features of oblique DSWs are
% investigated, namely (\emph{i}) relations between velocities and
% densities, and (\emph{ii}) absolute and convective instabilities
% regimes.

% The kinetic relations between upstream and downstream velocities and

To address these challenges, we \emph{asymptotically} determine the
spectrum of transverse perturbations to shallow but finite amplitude
NLS dark solitons beyond the Kadomtsev-Petviashvili (KP) limit.  This
enables determination of the maximum growth rate and associated
wavenumber of unstable perturbations.  Using adjoint methods, we
introduce a simple, accurate method for \emph{computating the
  spectrum} and its derivatives numerically for arbitrary soliton
amplitudes.  \emph{Simplified criteria} for the determination of the
\emph{separatrix} between absolute and convective instabilities are
derived.  The separatrix is determined in terms of the critical Mach
number $\Mcr$ as it relates to the soliton far field flow.  Oblique
dark solitons are convectively unstable when $M \ge \Mcr$ and
absolutely unstable otherwise.  Using our asymptotic and numerical
computations of the spectrum, we determine $\Mcr$, demonstrating that
$1 < \Mcr \lessapprox 1.4374$ with $\Mcr$ a monotonically increasing
function of soliton amplitude.

The oblique DSW trailing edge is well-approximated by an oblique dark
soliton.  In this study, we apply the soliton stability results to the
oblique DSW trailing edge in the stationary and nonstationary cases.
Stationary oblique DSWs result from the solution of a boundary value
problem (supersonic corner flow) while the nonstationary case arises
in the solution of a Riemann initial value problem.  We find that
oblique DSWs with supersonic downstream flows can be absolutely
unstable in contrast to classical oblique shocks.  We also show that
stationary and nonstationary oblique DSWs have the same downstream
flow properties in the shallow and hypersonic regimes.

We consider the (2+1)-dimensional defocusing (repulsive) nonlinear \SH
/ \\ \GP (NLS / GP) equation
\begin{equation}
  \label{eq:NLS}
  i \psi_t = -\frac{1}{2}\left( \psi_{xx} + \psi_{yy} \right)
  + |\psi|^2\psi~, \quad (x,y)\in \R^2, \ \  t>0 ~,
\end{equation}
along with appropriate initial and/or boundary data.  Equation
(\ref{eq:NLS}) models matter waves in repulsive BECs and intense laser
propagation in optically defocusing (\textit{i.e.}, with normal
dispersion) media.  In the variables
\begin{equation}
  \label{eq:Madelung}
  \psi = \sqrt{\rho} e^{i \phi} ~, \quad (u,v) = \nabla \phi ~,
\end{equation}
Equation (\ref{eq:NLS}) can be recast in terms of the fluid-like
variables $\rho$ (density) and $(u,v)$ (superfluid velocity)
\begin{subequations}
  \label{eq:63}
  \begin{align}
    \label{eq:23}
    \rho_t + (\rho u )_x + (\rho v )_y &= 0 ~,\\
    \label{eq:24}
    u_t + u u_x + v u_y + \rho_x &= \frac{1}{4} \left (
      \frac{\rho_{xx} + \rho_{yy}}{\rho} - \frac{\rho_x^2 +
        \rho_y^2}{2 \rho^2} \right )_x ~,\\
    \label{eq:25}
    v_t + u v_x + v v_y + \rho_y &= \frac{1}{4} \left (
      \frac{\rho_{xx} + \rho_{yy}}{\rho} - \frac{\rho_x^2 +
        \rho_y^2}{2 \rho^2} \right )_y ~ .
  \end{align}
\end{subequations}
Note that eqs.~\eqref{eq:63} in the dispersionless regime (neglecting
the right hand sides) correspond to the classical shallow water
equations (Euler equations of gas dynamics with adiabatic constant
$\gamma = 2$) with the speed of sound $\sqrt{\rho}$
\cite{courant_supersonic_1948}.

The outline of this paper is as follows.  Section
\ref{sec:transv-inst-dark} discusses the spectrum of unstable
transverse perturbations of dark solitons with asymptotic resolution
of the maximum growth rate and associated wavenumber in the shallow
regime.  Using analytic properties of the spectrum, we recap the
derivation of the general criteria for absolute and convective
instabilities and for oblique solitons, we derive the simplified
criteria in Sec.~\ref{sec:conv-absol-inst}.  The separatrix $\Mcr$ is
determined.  We derive nonstationary oblique DSWs of arbitrary
amplitude and stationary oblique DSWs in the shallow regime, showing
the connection between their downstream flows in Sec.~\ref{sec:M-t-b}.
The stationary case is compared with (2+1)-dimensional numerical
simulations.  Convective and absolute instability of oblique DSWs is
described in terms of the separatrix for the trailing edge dark
soliton.  Our numerical methods are presented in
Sec.~\ref{sec:numerical-methods}.  Finally, Sec.~\ref{sec:discuss}
contains a discussion of the results and the applicability of our
methods to other nonlinear dispersive problems.

%%%%%%%%%%%%%%%%%%%%%%%%%%%%%%%%%%%%%%%%%%%%%%%%%%%%%%

\section{Transverse instability of dark solitons}
\label{sec:transv-inst-dark}

It is well-known that dark soliton solutions of~\eqref{eq:NLS} exhibit
an instability to perturbations of sufficiently long wavelength in the
transverse direction along the soliton plane
\cite{kuznetsov-turitsyn-81}.  The eigenvalue problem associated with
linearizing~\eqref{eq:NLS} about the dark soliton leads to the
dispersion relation for unstable perturbations.  Beyond demonstrating
the \emph{existence} of an instability, knowledge of the dispersion
relation for a range of wavenumbers yields important \emph{properties}
of the instability, such as the growth rate $\gmax$, the maximally
unstable wavenumber $\kmax$, and whether or not the instability is
convective or absolute.

An example numerical computation of the eigenvalues for the spectral
problem in eq.~(\ref{eq:eigen-problem}) is shown in
Fig.~\ref{fig:Omega_plot}.  Since exact expressions are not known,
asymptotic approaches leveraging the shallow dark soliton, KP limit
\cite{zakharov__1975,alexander_transverse_1997} and others
\cite{pelinovsky_self-focusing_1995,kamchatnov-pitaevskii-08} have
been devised.  In this section, we complement these results by
determining the next order correction to the dispersion relation for
shallow dark solitons resulting in accurate approximation across a
wider range of soliton amplitudes.  We use this to determine $\gmax$
and $\kmax$ asymptotically.  These calculations are verified
numerically.

\fig{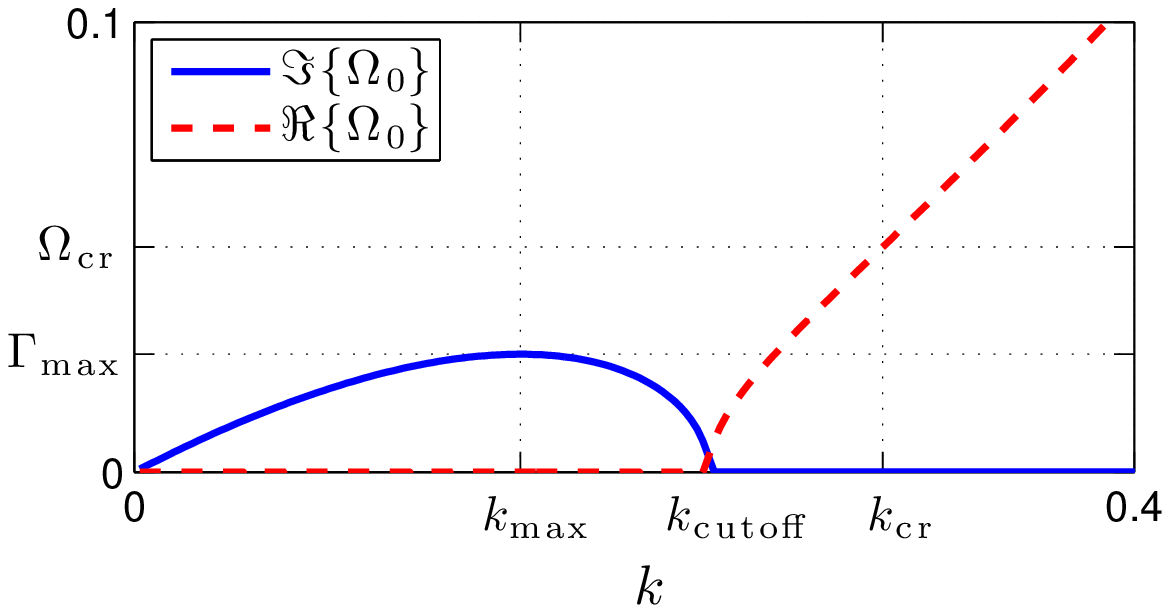}{Omega_plot}{0.7}{ The real (dashes) and imaginary
  (solid) parts of the discrete eigenvalue $\Omega_0(k;\nu)$ of the
  linearized NLS equation (spectral problem (\ref{eq:eigen-problem}))
  as functions of $k$ for $\nu=0.5\,$.  Delineated on the axes are: i)
  the cutoff wavenumber $\kcut$ [Eq.~\eqref{eq:kcut}], where the
  eigenvalue transitions from purely imaginary to real, ii) the
  maximal growth wavenumber and growth rate $(\kmax,\gmax)$
  [Eq.~\eqref{eq:max-growth}], and iii) the critical wavenumber and
  associated eigenvalue $(\kcr,\Ocr)$ [Eq.~\eqref{eq:k_cr2}]
  corresponding to the transition between absolute/convective
  instabilities.  }

\subsection{Dark Soliton}
\label{sec:dark-solitons}

Up to spatio-temporal shifts and an overall phase, the most general
line dark soliton solution of~\eqref{eq:NLS} is
\begin{align}
  \label{eq:15}
  \psi'_s(x',y',t') = &~\sqrt{ \rho} \left \{ \cos \phi + i \sin \phi
    \tanh \left [ a(\sin \beta x' - \cos \beta y'
      - v t' )\right ] \right \} \\
  \nonumber &~ \times \exp \left \{ i \left [ cx' + dy' - \left (
        \frac{c^2+d^2}{2} + \rho \right ) t' \right ] \right \}, \\
  \nonumber a = &~\sqrt{\rho} \sin \phi, \quad v = c \sin \beta - d
  \cos \beta - \sqrt{\rho} \cos \phi , \quad \phi \in [0,\pi]~,
\end{align}
where $\rho$ is the background density and the phase jump across the
soliton $2\phi$ determines the depression amplitude as
$\sqrt{\rho}|\sin \phi|$.  The soliton is propagating at an angle
$\beta$ with respect to the (horizontal) $x'$ axis, with horizontal
and vertical flow velocities $c$ and $d$, respectively.  Interpreting
this solution in the fluid context with density $|\psi'_s|^2$ and flow
velocity $\nabla \arg \psi'_s$, the soliton is a localized density
depression on a uniformly flowing background.  The Mach number of the
background flow is the total flow velocity divided by the speed of
sound
\begin{equation}
  \label{eq:5}
  M = \sqrt{\frac{c^2 + d^2}{\rho}}~.
\end{equation}
The soliton has the far field behavior
\begin{align*}
  \sin \beta x' - \cos \beta y' &\to \pm \infty ~, \\
  \psi'_s(x',y',t') &\to \sqrt{\rho} \exp \left \{ \pm i \phi +
    i \left [ cx' + dy' - \left ( 
        \frac{c^2+d^2}{2} + \rho \right ) t' \right ]  \right\} ~.
\end{align*}
Thus, five parameters determine the soliton uniquely, \ie $\rho, \phi,
\beta, c, d\,$.

Using the invariances of Eq.~\eqref{eq:NLS}
associated with rotation, Galilean transformation,  scaling, and phase,
we apply the coordinate transformation
\begin{align}
  \psi_s(x,y,t) =&~ \frac{-i}{\sqrt{\rho}} e^{-i \left [ (c \sin \beta
      - d \cos \beta)\frac{x}{\sqrt{\rho}} + (c \cos \beta + d \sin
      \beta ) \frac{y}{\sqrt{\rho}} + \frac{(c^2 + d^2)t}{2\rho}
    \right ]} \times 
  \\ \nonumber 
  \psi'_s&\left ( \frac{1}{\sqrt{\rho}}
    ( \sin \beta x + \cos \beta y ) + \frac{c}{\rho} t,
    \frac{1}{\sqrt{\rho}} ( - \cos \beta x + \sin \beta y ) +
    \frac{d}{\rho} t, \frac{t}{\rho} \right ) ~,
\end{align}
leading to the one-parameter family of dark solitons
\begin{equation}
  \label{eq:17}
  \psi_s(\xi,y,t;\nu) = \left [ i \kappa + \nu \tanh \left ( \nu \xi \right ) \right] e^{-it}, 
  \quad \nu^2 + \kappa^2 = 1 ~, 
\end{equation}
where $\nu = |\sin \phi| \in (0,1]$ and the frame moving with the
soliton is
$$
\xi \doteq  x - \kappa t~.
$$ 
The soliton amplitude is $\nu$.  When $\nu \ll 1$ the dark soliton is
in the shallow amplitude regime.  The soliton speed is $\kappa = -
\cos \phi = \sqrt{1-\nu^2}$.

\subsection{Linearized eigenvalue problem}
\label{sec:eigenvalue-problem}

To study the transverse instabilities of the dark
soliton~\eqref{eq:17}, we consider the ansatz for Eq.~\eqref{eq:NLS}
\begin{equation*}
  \psi(\xi,y,t;\nu) = \left [ \psi_s(\xi,y,t) e^{it} + 
    \varphi_R(\xi,y,t) + i\varphi_I(\xi,y,t)  \right ] e^{-it}~,
\end{equation*}
where $\varphi_R$, $\varphi_I$ are the real and imaginary parts of a
small perturbation.  Linearizing~\eqref{eq:NLS} results in the system
\begin{align} 
  \nonumber
 \frac{\partial}{\partial t} \vphib &= {\cal L}
  \vphib ~ , \\
  \label{eq:19}
  \vphib &\doteq
  \begin{bmatrix}
    \varphi_R \\
    \varphi_I
  \end{bmatrix} ~, \\[3mm]
  \nonumber {\cal L} &\doteq
  \begin{bmatrix}
    \kappa \partial_\xi + 2 \nu \kappa \tanh(\nu \xi ) & -\frac{1}{2}
    (\partial_{\xi\xi} + \partial_{yy}) - \nu^2[2 + \sech^2(\nu \xi)] \\*[2mm]
    \frac{1}{2}(\partial_{\xi\xi} + \partial_{yy}) - \nu^2[2 - 3\ \sech^2(\nu\xi)] 
    & \kappa\partial_\xi - 2 \nu \kappa \tanh(\nu \xi)
  \end{bmatrix}~.
\end{align}
It is expedient to decompose the perturbation as
\begin{equation}
  \label{eq:f}
  \vphib(\xi,y,t) = \frac{1}{2\pi} \int_\R {\bf f}(\xi;k) 
  e^{i[k y -  \Omega(k) t)]} dk ~, \quad 
  {\bf f}(\xi;k)  \doteq   \begin{bmatrix}   f_1(\xi;k)  \\
    f_2(\xi;k) \end{bmatrix}~.
\end{equation}
Substituting~\eqref{eq:f} into~\eqref{eq:19}  yields
the linearized spectral problem 
\begin{equation}
  \label{eq:eigen-problem}
 J L  \mathbf{f}(\xi;k) \ = \ -i\Omega(k) \mathbf{f}(\xi;k)~,
\end{equation}
where 
\begin{equation}
  \label{eq:L}
  L \  \doteq  \  L_0 + \frac{1}{2} k^2, \quad
  J \ \doteq \ \begin{bmatrix}
    0 & 1 \\
    -1 & 0
  \end{bmatrix}~,
\end{equation}
and
\begin{equation}
  \label{eq:L0}
  L_0 \ \doteq  \ \begin{bmatrix} -\frac{1}{2}\partial_{\xi\xi} + \nu^2[2 -
    3\sech^2(\nu\xi)] & -\kappa\partial_\xi + 2 \nu \kappa
    \tanh(\nu\xi) \\*[2mm]
    \kappa \partial_\xi + 2 \nu \kappa
    \tanh(\nu \xi ) &
    -\frac{1}{2} \partial_{\xi\xi} - \nu^2[2 + \sech^2(\nu \xi)]
  \end{bmatrix} ~.
\end{equation}
For $k \in \R$, $L_0$ and $L$ are self-adjoint with respect to the
$L^2(\R)$ inner product
\begin{equation}
  \label{eq:inner}
  \inner{\mathbf{g}}{\mathbf{h}} \ \doteq \ \int_{\mathbb R} \mathbf{g}^T
  \mathbf{h}^* \, d\xi~.
\end{equation}
For small $k$ it was shown formally in \cite{kuznetsov-turitsyn-81}
that: (\emph{i}) a double eigenvalue $\Omega(0) = 0$ bifurcates into
two distinct branches with each in $i\R$; (\emph{ii}) there is another
zero eigenvalue at the cutoff wavenumber
\begin{equation}
  \label{eq:kcut}
  \kcut \doteq \sqrt{\nu^2 -2 + 2\sqrt{\nu^4 - \nu^2 + 1}}~, \quad 
  \Omega(\kcut) = 0 ~. 
\end{equation}
These calculations were made rigorous in~\cite{rousset_simple_2010}
and can be summarized as follows.
\begin{theorem}[Rousset, Tzvetkov \cite{rousset_simple_2010}]
  \label{theo:spectrum}
  For $k \in (-\kcut,\kcut) \setminus \{0\}$, the
  system~\eqref{eq:eigen-problem} has exactly two purely imaginary
  eigenvalues which are simple and come in pairs $\pm \Omega_0(k)$.
  Therefore, the dark soliton is unstable to sufficiently long
  wavelength transverse perturbations.  Furthermore, for $k \in \R$,
  $|k| > \kcut$, the spectrum $\Omega(k)$ is real.
  % For $k\in \C$, $k\notin \{0,\pm\kcut\}$, the
  % system~\eqref{eq:eigen-problem} has exactly two eigenvalues, which
  % are simple and come in pairs $\pm \Omega_0(k)$.  $\Omega_0(k)$ is
  % purely imaginary for $k \in (-\kcut,\kcut) \setminus \{0\}$ and
  % $\Omega_0(k)$ is real (stable) for $k > \kcut$, or $k < -\kcut$.
  % For $k \in \C$ and the branch cut $(-\infty,0]
  % \cup [\kcut,\infty)$, $\Omega_0(k^*) = \Omega_0^*(k)$ for $\Re\{k\}
  % \in (0,\kcut)$ and $\Omega_0(k^*) = -\Omega_0^*(k^*)$ for $\Im\{k\}
  % \in [0,\kcut]^c$.
\end{theorem}

For the study of convective/absolute instabilities, knowledge of the
stable portion of the spectrum when $|k| > \kcut$ is required.  Based
on numerical and asymptotic computations, we conjecture the following.
\begin{conj}
  \label{conj:eigs}
  For $|k| > \kcut$, there exist exactly two real, simple eigenvalues
  $\pm \Omega_0(k)$.
\end{conj}

This conjecture is a natural extension of Thm.~\ref{theo:spectrum}.
See Appendix~\ref{ap:spectrum} for further details and comments.

Without loss of generality, we choose $\Omega_0(k)$ such that $\Im
\{\Omega_0(k)\} > 0$ for $0<k<\kcut$ and $\Re\{\Omega_0(k)\} > 0$ for
$k>\kcut$.  Thus, $\Omega_0(k)$ is the \emph{dispersion relation} for
transverse perturbations of the dark soliton~\eqref{eq:17}.  By
suitable choice of a branch cut, the eigenvalue $\Omega_0(k)$ can be
analytically continued for $k \in \C \setminus \{0,\pm \kcut \}$ with
$0$ and $\pm \kcut$ square root branch points.  We denote the growth
rate as
\begin{equation}
  \label{eq:Gamma}
  \Gamma(k) \doteq \Im\{\Omega_0(k) \}~,
\end{equation}
and the eigenfunction associated with $\Omega_0(k)$ as
\begin{equation*}
  \fzero(\xi;k)
  = 
  \begin{bmatrix}
    f_{0,1}(\xi;k) \\
    f_{0,2}(\xi;k)
  \end{bmatrix} ~.
\end{equation*}

In Section~\ref{sec:numerical-methods} we discuss our numerical method
for computing $\Omega_0(k)$ for $k \in \C$.  To illustrate the
spectrum, Figure~\ref{fig:Omega_plot} presents the dependence of the
(real or imaginary) eigenvalue, $\Omega_0(k)$, on $k$.  Figures
\ref{fig:eigenstuff_imag} and~\ref{fig:eigenstuff_real} present the
computed continuous and discrete spectra for particular wavenumbers $0
< k<\kcut$ ($\Omega_0\in i\R$) and $k>\kcut$ ($\Omega_0\in\R$) as well
as the associated localized eigenfunctions.  Note that the
eigenfunctions are neither symmetric nor anti-symmetric.

\fig{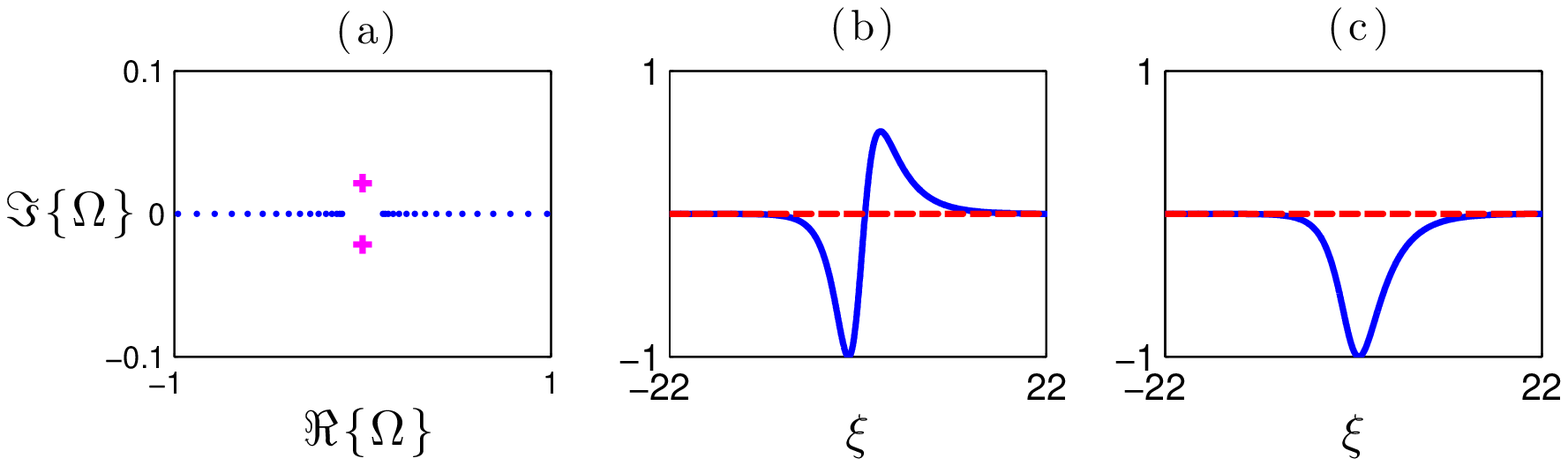}{eigenstuff_imag}{1}{ (a) Numerical approximation
  of the continuous spectrum ({\bf {\color{blue}$\bullet$}}) and the
  two purely imaginary discrete eigenvalues $\pm\Omega_0 \approx \pm
  0.022i$ ({\bf {\color{magenta} +}}) computed for the linearized
  system~\eqref{eq:eigen-problem} with $\nu=0.5$ and
  $k=0.2<\kcut\approx 0.23$ [Eq.~\eqref{eq:kcut}].  
  The real (solid) and imaginary (dashed) parts of the corresponding two
  component localized eigenfunction are shown in 
  (b) $f_{0,1}(\xi)$ and (c) $f_{0,2}(\xi)$. 
}

\fig{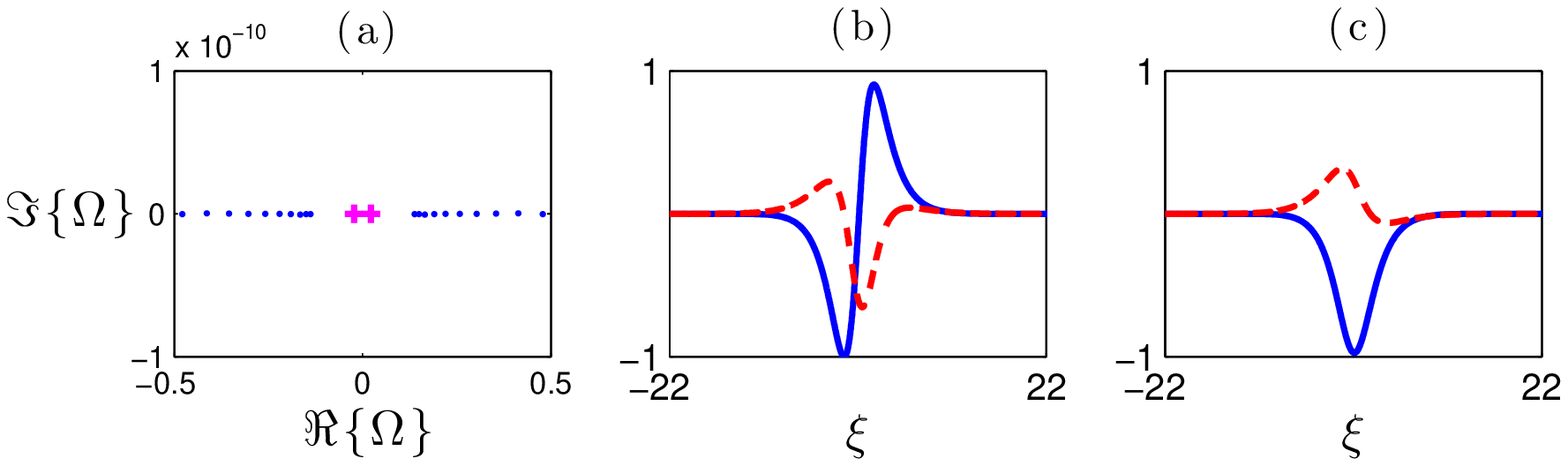}{eigenstuff_real}{1}{ Same as for
  Fig.~\ref{fig:eigenstuff_imag} except $k=0.25>\kcut$ and
  $\pm\Omega_0 \approx \pm 0.0215$.  }

\subsection{Asymptotic eigenvalue}
\label{sec:asympt-eigenvalue}

It follows from~\eqref{eq:kcut} that for shallow amplitude solitons,
$0 < \nu \ll 1$, the cutoff wavenumber is small, \textit{i.e.}, $\kcut
\sim \frac{\sqrt{3}}{2}\nu^2$.  In
Appendix~\ref{sec:numer-meth-cauchy} we prove:
\begin{proposition}
  \label{prop:asympt-eig}
  For shallow amplitude, $0 < \nu \ll 1$, and either $k < \kcut$ or $
  \kcut < k \sim \mathcal{O}(\nu^2) \ll 1$, the eigenvalue for
  (\ref{eq:eigen-problem}) satisfies
  \begin{eqnarray}
    \label{eq:asympt-eig}
    \Omega_0(k) = \underbrace{\frac{k}{3} \sqrt{2 \sqrt{3} k -
        3\nu^2}}_{\rm KP} \ + \ \underbrace{\frac{k^2(\sqrt{3} \nu^2 -
        k)}{6 \sqrt{2\sqrt{3}k - 3 \nu^2}} }_{\rm NLS~correction} \ + \
    \mathcal{O}(k^{7/2})~,
  \end{eqnarray}
  where the first (leading order) term is the dispersion relation for
  the KP equation and the second term is the $\mathcal{O}(k^{5/2})$
  correction arising from the NLS equation.
\end{proposition}

Equation \eqref{eq:asympt-eig} gives an asymptotic approximation to
the eigenvalue for long wave perturbations of shallow dark solitons.
The dispersion relation for the KP equation is well known
(cf.~\cite{zakharov__1975,alexander_transverse_1997}).  The new
$\mathcal{O}(k^{5/2})$ correction term enables us to accomplish the
following.
\begin{itemize}
\item Implement an accurate, explicit calculation of the maximum
  growth rate and associated wavenumber of unstable perturbations
  (Sec.~\ref{sec:calc-maxim-growth}).
\item Show that the separatrix between absolute and convective
  instabilities is \emph{supersonic}
  (Sec.~\ref{sec:asympt-separ-crit}).
\item Validate the numerical computations of $\Omega_0(k)$, which are
  sensitive and computationally demanding, especially in the shallow
  regime.
\end{itemize}

\subsection{Calculation of the maximum growth rate}
\label{sec:calc-maxim-growth}

The maximal growth wavenumber $\kmax$ and the maximum growth rate
$\gmax$ are defined by
\begin{equation}
  \label{eq:max-growth}
  \Omega_0'(\kmax)=0~, \quad \gmax = \Im\{\Omega_0(\kmax)\}~.
\end{equation}
Since $\Omega_0(k)$ is real for $k>\kcut$ it follows that $\kmax <
\kcut$ (see Fig.~\ref{fig:Omega_plot}).  Using
Proposition~\ref{prop:asympt-eig} we find
\begin{cor}
\label{cor:max-growth}
\begin{align}
  \label{eq:59}
  \kmax &= \underbrace{\frac{\nu^2}{\sqrt{3}}}_{\mathrm{KP}} +
  \underbrace{\frac{5 \nu^4}{18 \sqrt{3}}}_{\mathrm{NLS~correction}}
  + \mathcal{O}(\nu^6) ~, \\
  \label{eq:43}
  \gmax &= \underbrace{\frac{\nu^3}{3 \sqrt{3}}}_{\mathrm{KP}} +
  \underbrace{\frac{\nu^5}{9\sqrt{3}}}_{\mathrm{NLS~correction}} +
  \mathcal{O}(\nu^7) ~.
\end{align}
\end{cor}
The proof follows by expanding $\kmax$ and $\gmax$ for small $\nu$ and
solving eq.~\eqref{eq:max-growth} with the approximation
\eqref{eq:asympt-eig}.

A comparison of these results with numerical computations (discussed
in Sec.~\ref{sec:numerical-methods}) is shown in
Fig.~\ref{fig:gmax_kmax}.  The computations exhibit excellent
agreement with the asymptotics as well as the expected scaling of the
errors with $\nu$.

\fig{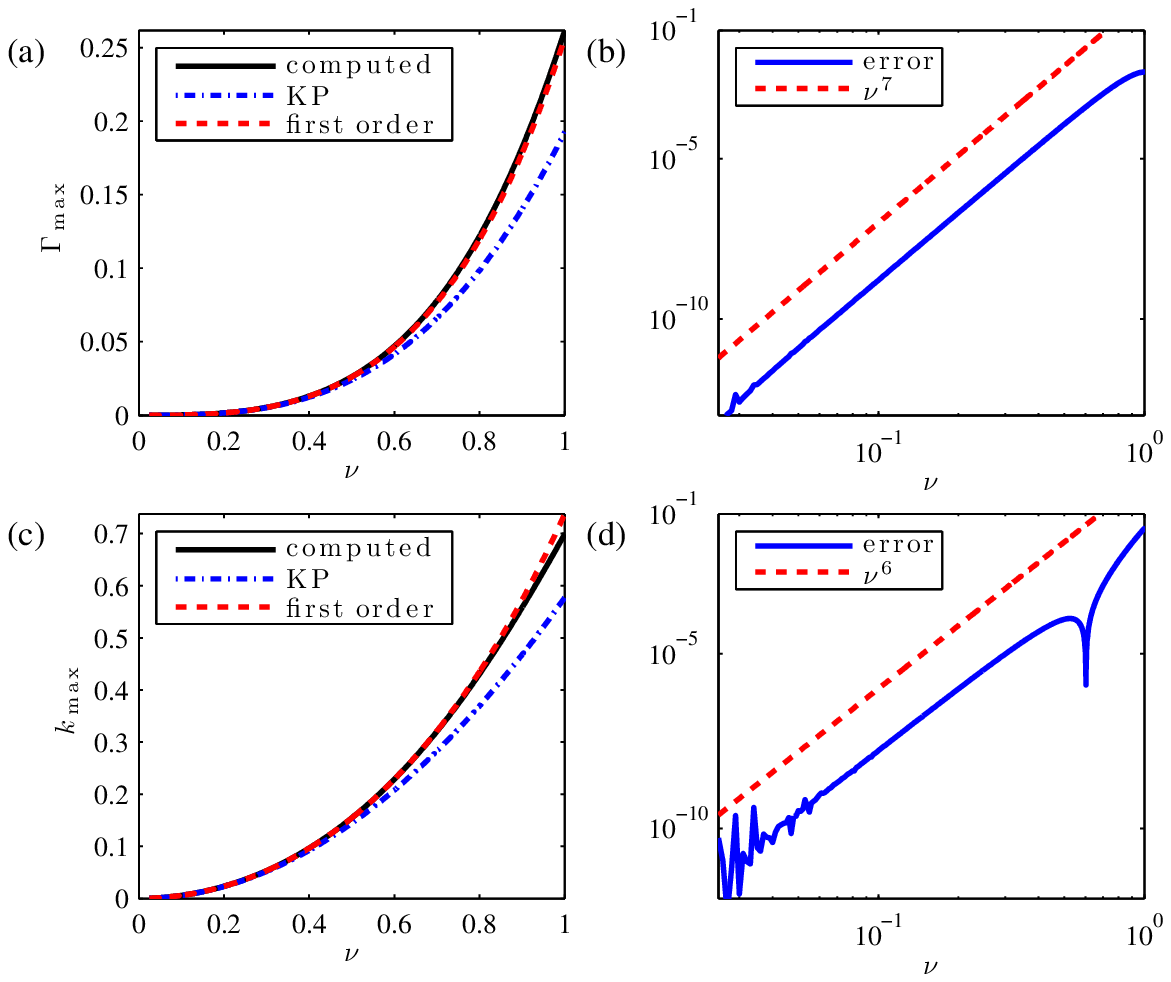}{gmax_kmax}{1} {Numerically computed maximum growth
  rate $\gmax$ (a) and maximally unstable wavenumber $\kmax$ (c) as
  functions of $\nu$ for dark solitons of the NLS
  equation~\eqref{eq:NLS}. The KP limit and its first order correction
  are presented for comparison.  Plots (b) and (d) are the
  corresponding differences between the highly-accurate computed
  values and asymptotic approximations \eqref{eq:59}--\eqref{eq:43},
  exhibiting the expected scaling with $\nu$.  }

%%%%%%%%%%%%%%%%%%%%%%%%%%%%%%%%%%%%%%%%%%%%%%%%%%%%%%

\section{Convective and absolute instabilities of dark solitons}
\label{sec:conv-absol-inst}

We begin by reviewing the notions of absolute/convective instabilities
and the general criteria for distinguishing between them.  For more
detailed discussions
see~\cite{sturrock-58,briggs-64,pitaevskii_lifshitz_kinetics_1981,infeld-rowlands-1990,
  schmid_henningson_2001}.  Qualitatively, absolute and convective
instabilities can be defined as follows (see illustration in
Fig.~\ref{fig:instabil_fig}).
\begin{defin}
  \label{def:abs-conv}
  A solution is said to be {\bf absolutely unstable} if generic,
  small, localized perturbations grow arbitrarily large in time at
  each fixed point in space.  A solution is said to be {\bf
    convectively unstable} if small, localized perturbations grow
  arbitrarily large in time but decay to zero at any fixed point in
  space.
\end{defin}

It is important to note that Definition~\ref{def:abs-conv} depends
implicitly on the reference frame as can be gleaned from
Fig.~\ref{fig:instabil_fig} where panel (b) is a rotation in the
$x$-$t$ plane of panel (a).  Such a rotation implies that the observer
in (b) is moving faster to the left than the observer in (a).  Thus,
if the observer ``outruns'' the growing perturbation, then the
instability is convective.  Equivalently, if the background flow speed
is faster than the expanding, unstable perturbation, and after
sufficient time passes the solution returns to is unperturbed state,
the instability is convective.

\fig{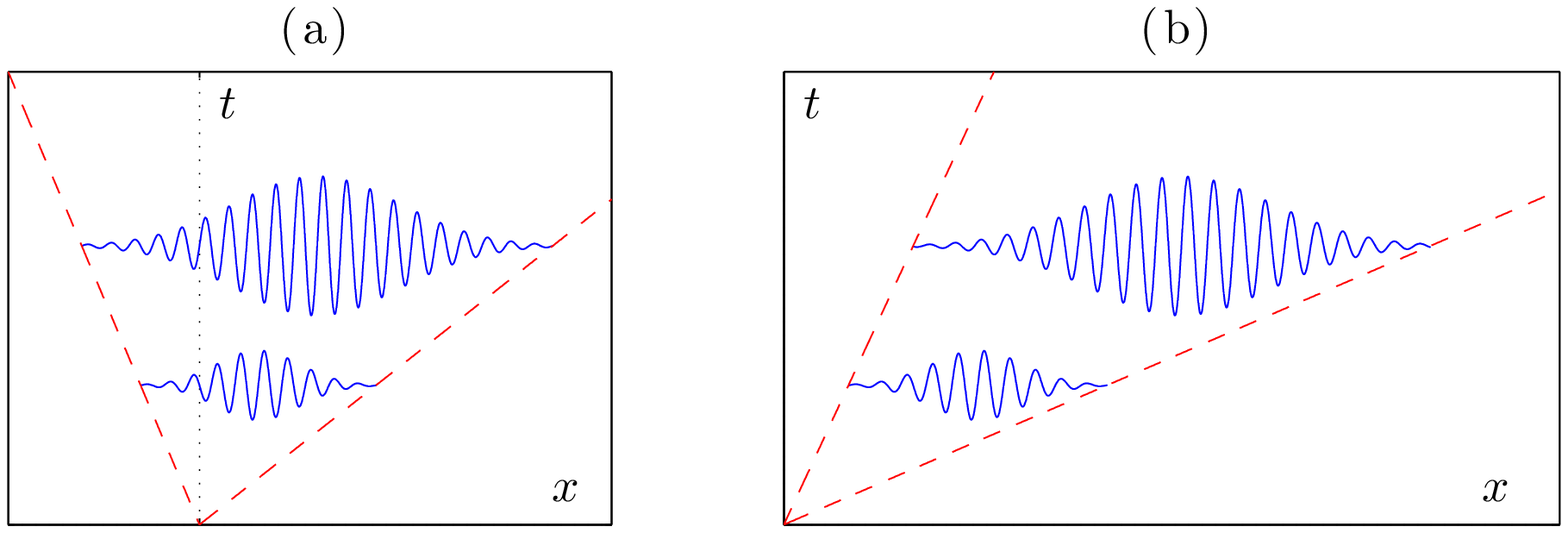}{instabil_fig}{1} {Illustration of (a) absolutely
  and (b) convectively unstable waves.}

\subsection{Review of the general criteria for distinguishing between
  instabilities}

Absolute and convective instabilities can be distinguished
analytically.  Consider an initial value problem on the entire line,
\ie a (1+1)-dimensional linearized system on $(x,t)\in \R\times
(0,\infty)$.  The usual approach for studying instabilities is to
consider a small, spatially extended {\bf plane wave} perturbation
$e^{i(kx - \omega t)}$ of some wavenumber $k$ and corresponding
frequency $\omega = \Omega(k)$ determined by a zero of the dispersion
function $\cD(\omega,k) = 0$.  The zero state is {\bf stable} if and
only if $\Im\{\Omega(k)\}\le 0$ for all zeros of the dispersion
function.  However, the evolution of a particular, localized
perturbation involves a Fourier integral over all real wavenumbers so
that treating a single wavenumber is insufficient to fully describe
any instabilities observed (or not observed) in a physical
system~\cite{sturrock-58}.

The resolution calls for a different approach to instability analysis.
Instead of a plane wave perturbation one assumes that the system is
perturbed by a localized {\bf impulse}, \ie a Dirac delta function at
position $x$ and $t=0$.  In this case the solution is the Green's
function
\begin{equation}
  \label{eq:G}
  G(x,t) \ = \ \iint \frac{ e^{i (kx - \omega t)}  }{\cD(\omega,k)}\,
  d\omega \, dk~,
\end{equation}
where the Fourier integral is carried out over real wavenumbers and
the Bromwich frequency contour lies above all zeros of
$\cD(k,\omega)$.  In connection with the plane-wave analysis, the
system is unstable if and only if the solution grows without bound
along \emph{some} reference frame, \ie there is a velocity $V$ such
that for fixed $x$,
$$
 G(x-Vt,t) \stackrel{t\to\infty}{\longrightarrow} \infty  \ \iff \
 \mbox{unstable}~. 
$$
However, when considering a \emph{particular} reference frame, say
$\xi=x-V_0t$ for fixed $V_0$, if the solution grows without bound
(resp. decays to zero) at a certain fixed point in space, $x$, then
the system is absolutely (resp. convectively) unstable in this
reference frame, \ie
\begin{itemize}
\item $G(x - V_0t,t) \stackrel{t\to\infty}{\longrightarrow} \infty
  \iff$ absolutely unstable.
\item $G(x - V_0t,t) \stackrel{t\to\infty}{\longrightarrow} 0 \iff$
  convectively unstable.
\end{itemize}

Exponential integrals of the type in~\eqref{eq:G} have two competing
effects.  Zeros of the dispersion function $\omega = \Omega(k)$ can
lead to exponential growth when $\Im\{\Omega(k)\} > 0$ or cancellation
and decay due to rapid oscillation when $\Im\{\Omega(k)\} = 0$ for
large $t$.  To ascertain whether the system is absolutely or
convectively unstable one needs to discover which of these opposite
tendencies dominates.  A number of methods for distinguishing between
convective and absolute instabilities have been suggested, dating back
to the work of Sturrock \cite{sturrock-58} andBriggs~\cite{briggs-64}.
See also~\cite{bers-83,huerre-monkewitz-85,brevdo-bridges-96}.  For
completeness we outline the general criteria below.

Here we assume that $\cD(\omega,k)$ is known explicitly.  The
$\omega$-integral in~\eqref{eq:G} is along a contour that lies above
all the zeros of $\cD(\omega,k)$ for each fixed, real $k$ and we
further assume that $\cD(\omega,k)$ is entire in $(\omega,k)$ above
this contour.  Hence, for $t>0$ the $\omega$-integral may be carried
out by closing the contour in the lower half-plane and summing over
the residues of the dispersion function expanded at each of its
roots. Assuming the roots of $\cD(\omega,k)$ are simple (multiple
roots do note pose a serious difficulty~\cite{infeld-rowlands-1990}),
the resulting integral can be written as
\begin{equation}
  \label{eq:G-sum} 
  G(x,t) \ = \ -2\pi i \sum_{n=1}^N \int_{-\infty}^\infty \frac{ e^{i (kx -
      \Omega_n(k) t)}  }{\cD'(\Omega_n(k),k) }\, dk~, 
\end{equation}
where the sum is over the $N$ zeros of the dispersion function
$$
\cD(\Omega_n(k),k) = 0~, \quad \cD'(\omega,k) \doteq
\pd{\cD(\omega,k)}{\omega}~.
$$

The problem is to determine the long time behavior of \eqref{eq:G-sum}
for which the method of steepest descent is applicable
(cf.~\cite{miller_applied_2006}).  For this, we restrict ourselves to
the point moving with speed $V$, $x = Vt$.  Then, by suitable
deformation of the real line to the steepest descent contour, the
dominant contributions arise from the saddle points of the exponent
satisfying
\begin{equation*}
  \label{eq:71}
  \dd{}{k} \Omega_n(k_{n,m})= V~, \quad n = 1, \ldots N, \quad m = 1,
  \ldots, M_n ,
\end{equation*}
allowing for multiple saddle points along each branch of the
dispersion relation.  Note that the zeros of $\cD'$, double roots of
the dispersion function, do not contribute appreciably to the integral
because they cancel in the sum \eqref{eq:G-sum}.  Using the method of
steepest descent, one recovers the dominant long time behavior
\begin{equation*}
  \label{eq:72}
  G(Vt,t) \sim \mathcal{O}(e^{\gamma_\mathrm{max} t}/\sqrt{t} ) ~, \quad t \to
  \infty ~, 
\end{equation*}
where 
\begin{align*}
  \gamma_\mathrm{max}(V) \doteq \max_{n, m} \Im \left \{
    \Omega_n(k_{n,m}) - V k_{n,m} \right \} ~.
\end{align*}
Thus, if $\gamma_\mathrm{max} > 0$, then an impulse perturbation at $t
= 0$, $x = 0$ grows without bound along the line $x = Vt$ and the
instability is absolute.  Otherwise, if $\gamma_\mathrm{max} \le 0$,
the perturbation decays along the line $x = Vt$ and so the instability
is convective.  These have been referred to as the {\bf Bers-Briggs
  criteria}~\cite{briggs-64,infeld-rowlands-1990}.

\subsection{Simplified criteria for the separatrix of soliton
  instabilities}
\label{sec:crit-separ-solit}

Many previous studies applied the general criteria for classifying
instabilities to dissipative systems (plasma physics, viscous fluids,
etc.) where the dispersion relation was known explicitly. Given
explicit (and sufficiently simple) dispersion relations, the analysis
of the stationary points can be carried out directly.  However, the
dispersion relation $\Omega(k)$ is unknown for dark solitons of the
NLS equation.  It can be computed numerically, but this makes the
analysis of saddle points in the complex-$k$ and / or complex-$\omega$
planes quite challenging.  Fortunately, as derived below, there are
simplified analytic criteria for the transition point between absolute
and convective instabilities of NLS solitons that rely solely on
computations of the dispersion relation for real $k$.

Using the Laplace transform in eq.~\eqref{eq:19}, the linearized
evolution of an initial $L^2(\R^2)$ perturbation $\vphib_0(\xi,y)$ to
the dark soliton satisfies
\begin{equation*}
  \label{eq:4}
  \vphib(\xi,y,t) = \frac{1}{2\pi} \int_{C_\mathrm{B}} e^{- i \omega t} \left (
    \mathcal{L} + i \omega \right )^{-1}
  \vphib_0(\xi,y) \, d\omega ,
\end{equation*}
where the Bromwich contour $C_\mathrm{B}$ lies above all eigenvalues
of $\mathcal{L}$.  In order to investigate the unstable
\emph{transverse} dynamics in $(y,t)$, we project onto the
eigenfunction $\fzero$ and perform the contour integration over
$C_\mathrm{B}$ resulting in the following representation of the
dynamics
\begin{align}
  \nonumber
  \vphib(\xi,y,t) &= \frac{-i}{2\pi} \int_0^\infty \frac{e^{i(ky -
      \Omega_0(k) t)}}{\Omega_0(k)} \fzero(\xi;k) \, dk ~ .
\end{align}
The integral is taken over $(0,\infty)$ by use of the invariance $k
\to -k$ of the eigenpair \\
$(\Omega_0(k),\fzero(\xi;k))$.

By performing a Galilean shift in the NLS equation~\eqref{eq:NLS} as
\begin{equation}
  \label{eq:105}
  \psi(x,y,t) \to \psi'(x,y,t) = e^{i(-w y - w^2 t /2)} \psi(x,y+wt,t) ~,
\end{equation}
the dispersion relation for transverse perturbations becomes 
\begin{equation}
  \label{eq:Omega_transverse}
  \Omega_0(k) \to \Omega(k) = -w k + \Omega_0(k) \, ,
\end{equation}
where $-w$ is the flow speed parallel to the plane of the dark
soliton~\eqref{eq:17}.  This is equivalent to investigating the
behavior of the perturbation in eq.~(\ref{eq:105}) along the line $y =
wt$.  With this substitution, we consider eq.~\eqref{eq:105} whose long
time asymptotic behavior requires the evaluation of
\begin{equation}
  \label{eq:8}
  I(t) = \int_0^{\kcut} \frac{ e^{-i \Omega(k) 
      t}}{\Omega_0(k)} \fzero(k) \, dk \, , \quad t \gg
  1 ~,
\end{equation}
where the dependence on $\xi$ is suppressed.  The integral over
$(\kcut,\infty)$ is negligible because the dispersion relation is
purely real (the stationary phase method yields algebraic decay in
$t$, cf.~\cite{miller_applied_2006}).  Introducing the change to a
complex variable $z = \Omega(k)$, eq.~\eqref{eq:8} becomes
\begin{equation}
  \label{eq:18}
  I(t) = \int_{C} \frac{e^{-i z t}}{\Omega_0(z) \Omega'(z)}
  \fzero(z) \, dz \, ,
\end{equation}
where $\Omega'(z) = -w + \Omega_0'(z)$ and the contour is $C \doteq
\{z = \Omega(k) ~ | ~ k \in [0,\kcut] \}$.  Two distinct possibilities
arise.  \fig{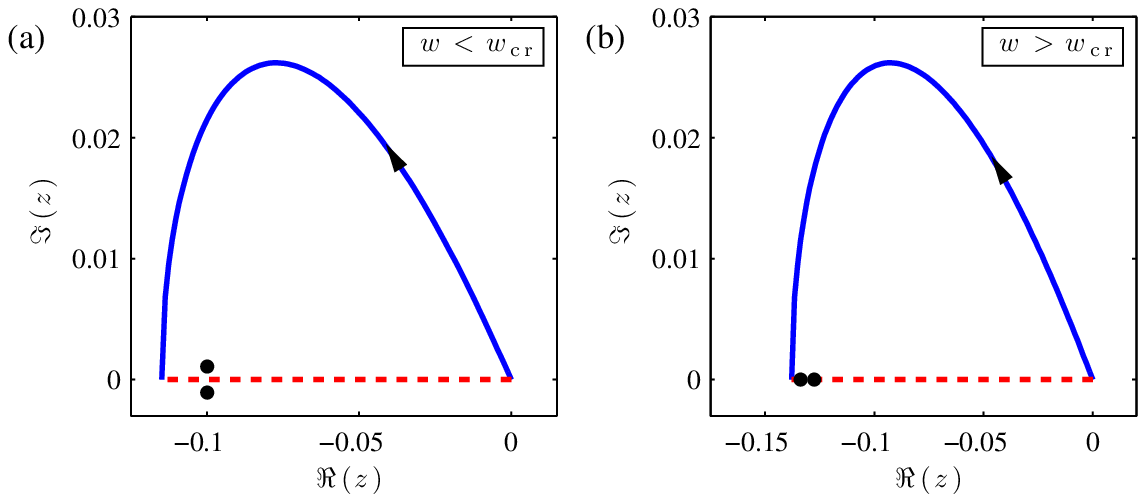}{contours}{1}{Integration contours $C$
  (solid curves) in the complex $z$ plane for eq.~(\ref{eq:18}) and
  the real interval $[-w\kcut ,0]$ (dashed lines).  The filled circles
  correspond to poles of the integrand where $\Omega'(k) =
  \Omega_0'(k) - w = 0$, $z = \Omega(k)$, which in (a) prevent the
  smooth deformation of $C$ to $[-w\kcut,0]$ giving rise to an
  absolute instability.  Parameter values are $\nu = 0.5$, $\wcr
  \approx 0.535$.  (a) $w = 0.5 < \wcr$.  (b) $w = 0.6 > \wcr$.  See
  also Fig.~\ref{fig:omega_prime}.}
\begin{enumerate}
\item A zero of $\Omega'$ gives a residue contribution to Cauchy's
  theorem when deforming $C$ to the real interval $[-w\kcut,0]$ as in
  Fig.\ref{fig:contours}(a).  In this case the integral diverges
  exponentially as $t\to \infty$ and the instability is
  \textbf{absolute}.
\item The zeros of $\Omega'$ do not lie between $C$ and the real line
  as in Fig.\ref{fig:contours}(b) (they may lie \emph{on} the real
  axis) so that there is a smooth deformation of the contour $C$ to
  the real interval $[-w\kcut,0]$.  In this case the integral decays
  to zero as $t\to \infty$ and the instability is \textbf{convective}.
\end{enumerate}
\fig{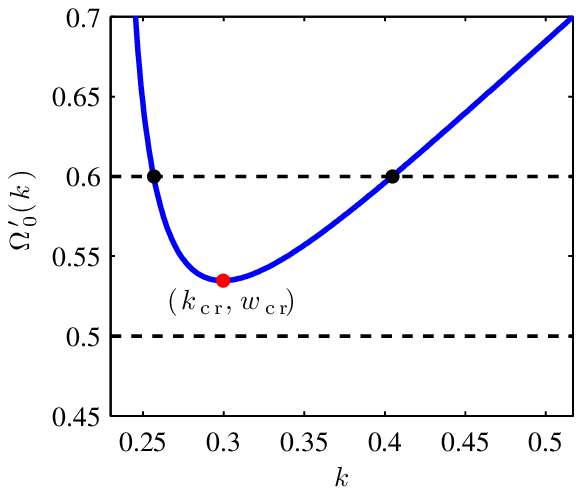}{omega_prime}{0.5}{Plot of $\Omega_0'(k)$ for real $k
  > \kcut \approx 0.230$ and $\nu = 0.5$.  The minimum of this curve
  corresponds to the coalescence of the poles in
  Fig.~\ref{fig:contours} and the critical transverse flow speed
  $\wcr$ at which the instability changes from absolute to convective.
  The dashed lines correspond to the values of $w$ used to compute
  Fig.~\ref{fig:contours}(a) (lower, absolute instability) and
  Fig.~\ref{fig:contours}(b) (upper, convective instability).}  

As discussed in the previous section, the saddle points $\Omega'(k_0)
= \Omega_0'(k_0) - w = 0$ give the long time asymptotic behavior
$\vphib \sim \mathcal{O}(e^{ i \Omega(k_0)t}/\sqrt{t})$, $t \to
\infty$. As the transverse flow speed $w$ is varied, the type of
instability changes from absolute to convective.  The transition from
absolute to convective instability occurs at $w = \wcr$ when two zeros
of $\Omega'(k)$ merge on the real line.  That is, they form a double
zero so that $\Omega''(\kcr) = \Omega_0''(\kcr) = 0$ and
$\Omega_0'(\kcr)$ is minimum.  This behavior is depicted in
Figs.~\ref{fig:contours} and \ref{fig:omega_prime} with
\ref{fig:contours}(a) showing two complex conjugate zeros for $w <
\wcr$ while \ref{fig:contours}(b) reveals their splitting into two
real zeros for $w > \wcr$.  These real zeros are depicted in
Fig.~\ref{fig:omega_prime} for $w > \wcr$. By an appropriate choice of
the branch cut, one can show that $\Omega_0(k) = \Omega_0^*(k^*)$ so
that complex zeros of $\Omega'(k) = \Omega_0'(k) -w$ come in conjugate
pairs.  This proves
\begin{proposition}
  \label{prop:simpl-crit-separ}
  The critical wavenumber $\kcr$ and critical transverse velocity
  $\wcr$ for the transition between absolute and convective
  instability are real.  They satisfy the \textbf{simplified criteria}
  \begin{subequations}
    \label{eq:k_cr1}
    \begin{align}
      \label{eq:inflection1} 
      \pdd{\Omega}{k}(\kcr;\wcr) & = 0~, \\
      \label{eq:stationary1}
      \pd{\Omega}{k}(\kcr;\wcr) & = 0~ .
    \end{align} 
  \end{subequations}
\end{proposition}
These conditions were first proposed in
\cite{kamchatnov-pitaevskii-08}.

\begin{corollary}
  \label{cor:simpl-crit-wcr}
  \begin{subequations}
    \label{eq:k_cr2}
    \begin{align}
      \label{eq:inflection2} 
      \pdd{\Omega_0}{k}(\kcr) &= 0~, \\
      \label{eq:stationary2}
      \wcr & = \pd{\Omega_0}{k}(\kcr)~.
    \end{align} 
  \end{subequations}
\end{corollary}

The proof follows from \eqref{eq:Omega_transverse}
and~\eqref{eq:k_cr1}.

When the transverse flow speed is subcritical, $w < \wcr$, the dark
soliton is absolutely unstable and when $w > \wcr$ the dark soliton is
convectively unstable.  The soliton family is parametrized by its
amplitude $\nu$, thus $\nu\mapsto\wcr(\nu)$ forms a {\bf separatrix}
between absolute and convective instabilities.  The separatrix
$\wcr(\nu)$ can also be interpreted as the speed at which an initially
localized perturbation spreads in time.  Thus a convective instability
occurs when the background flow speed, carrying the perturbation's
center of mass, exceeds the speed at which the perturbation spreads
out.

In general, the determination of $\wcr(\nu)$ via~\eqref{eq:k_cr2}
requires numerical computation.  Even so, Eqs.~\eqref{eq:k_cr2} are
much easier to use than the general criteria because the general
criteria depend on $\Omega_0(k)$ over the complex-$k$ plane whereas
\eqref{eq:k_cr2} only depends on $\Omega_0(k)$ for real $k$.

\subsection{The separatrix in the shallow amplitude regime}
\label{sec:asympt-separ-crit}
% \edit{Incorporate Boaz's details for the leading order, KP
%   calculations; this section is a bit jumbled right now.}

The shallow-amplitude asymptotics of the dispersion
relation~\eqref{eq:asympt-eig} enable us to explicitly compute
$\Omega_k$ and $\Omega_{kk}$, determine the critical wavenumber $\kcr$
and find the separatrix $\wcr(\nu)$ between absolute and convective
instabilities.  Here it is convenient to use the wavenumber scaling
(see Appendix \ref{sec:numer-meth-cauchy})
$$
k = \nu^2 p~.
$$
The asymptotic dispersion relation~\eqref{eq:Omega_transverse} becomes
\begin{align*}
  \Omega(p;w) &\sim -\nu^2 w p + \nu^3 \frac{p}{3} (2 \sqrt{3}p -
  3)^{1/2} + \nu^5 \frac{p^2(\sqrt{3} - p)}{6(2 \sqrt{3} p -3)^{1/2}}
  \\
  &\doteq -\nu^2 w p + \nu^3 \Lambda_0(p) + \nu^5 \Lambda_1(p) , \quad
  0 < \nu \ll 1 ~.
\end{align*}
The simplified criteria~\eqref{eq:k_cr2} give
\begin{align*}
  \kcr &= \nu^2 \pcr = \nu^2 (p_0 + \nu^2 p_1) \ + \ O(\nu^4) ~, \\
  \Omega_{kk}(\kcr) &\sim \frac{1}{\nu} \Lambda_0''(p_0) +
  \nu \left [ \Lambda_0'''(p_0) p_1 + \Lambda_1''(p_0) \right ] ~,\\
  \Omega_k(\kcr;\wcr) &\sim -\wcr + \nu \Lambda_0'(p_0) +
  \nu^3 \left [ \Lambda_0''(p_0)p_1 + \Lambda_1'(p_0) \right ] ~.
\end{align*}
Equating like coefficients of $\nu$ and using (\ref{eq:k_cr2}), yields
\begin{proposition}
  \label{prop:separ-shall-ampl}
  The first order asymptotic approximation of the critical velocity
  and wavenumber are
  \begin{subequations}
    \label{eq:51}
    \begin{align}
      \label{eq:13}
      \wcr &=  \underbrace{\nu}_{\mathrm{KP}} + \underbrace{\frac{2
          \nu^3}{9}}_{\mathrm{NLS~correction}} \ + \ O(\nu^5)~, \\
      \label{eq:11}
      \kcr &= \underbrace{\frac{2 \nu^2}{\sqrt{3}}}_{\mathrm{KP}}  +
      \underbrace{\frac{\nu^4}{3\sqrt{3}}}_{\mathrm{NLS~correction}} \ + \
      O(\nu^6)~ .
    \end{align}
  \end{subequations}
\end{proposition}
For comparison, Fig.~\ref{fig:wcr_pcr} shows the numerical solution of
the system (\ref{eq:k_cr2}) and the asymptotics in (\ref{eq:51}).  The
numerical details are presented in Sec.~\ref{sec:variational}.

\fig{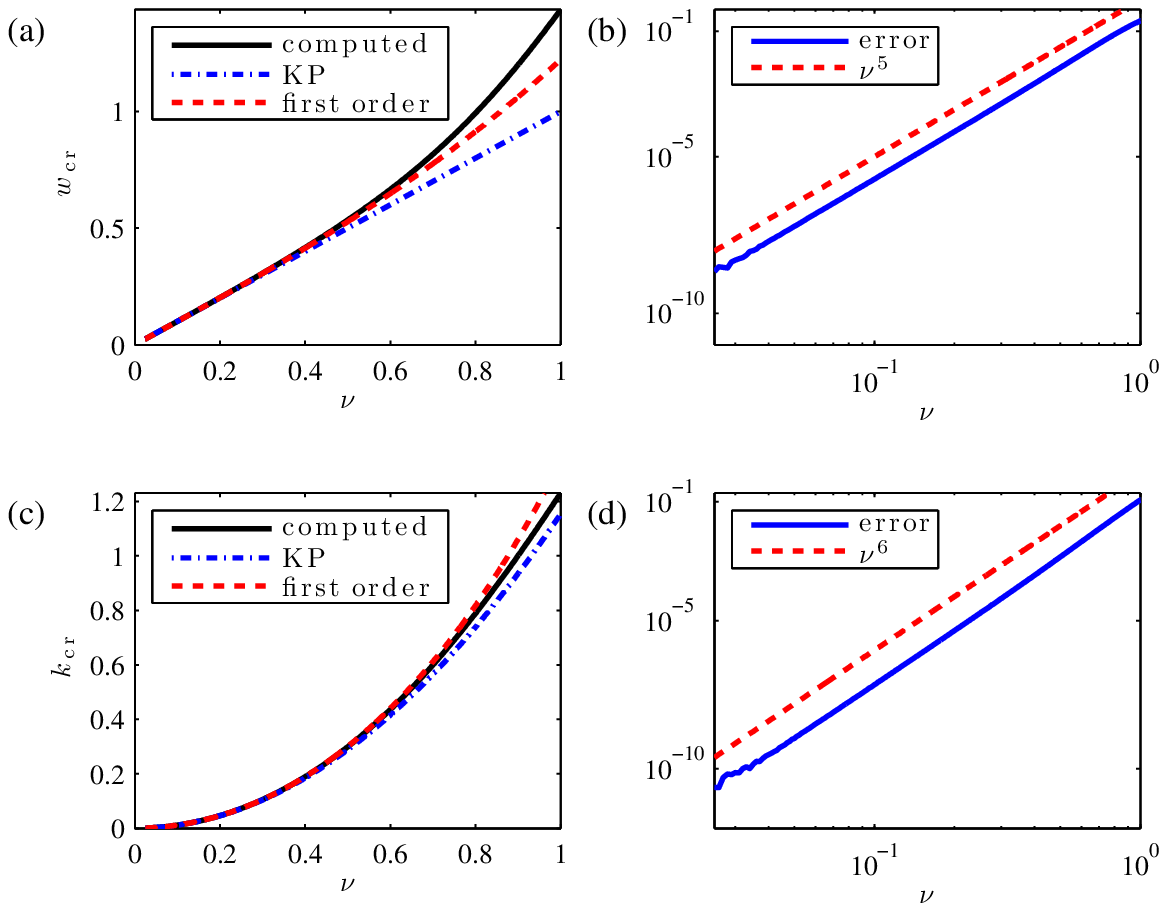}{wcr_pcr}{1}{(a) The separatrix $\wcr(\nu)$ between
  absolute and convective instabilities.  For speeds $w \ge \wcr$ ($w
  < \wcr$), the dark soliton is convectively (absolutely) unstable.
  (c) The critical wavenumber $\kcr$ satisfying the condition
  \eqref{eq:inflection1}.  (b) and (d) are the differences in the
  numerically computed values of $\wcr$, $\kcr$ and the first order
  approximations in \eqref{eq:13} and \eqref{eq:11}, respectively,
  showing the expected error scalings. }

\subsection{Convective/absolute instabilities of spatial dark solitons}
\label{sec:conv-inst-spat}

The natural reference frame for studying the convective or absolute
nature of soliton instabilities is the one moving with the soliton.
In this reference frame, both the soliton density and velocity are
independent of time.  The dark soliton is referred to as a
\textbf{spatial dark soliton}.  Such structures arise, for example, in
the context of flow past an impurity
\cite{el_oblique_2006,mironov_structure_2010}, flow over extended
obstacles, and dispersive shock waves
\cite{el_two-dimensional_2009,hoefer_theory_2009}.

The spatial dark soliton in \eqref{eq:15} satisfies $v = 0$, which
determines the phase jump
\begin{equation}
  \label{eq:1}
  \cos \phi = \frac{c \sin \beta - d \cos \beta}{\sqrt{\rho}} ~.
\end{equation}
This soliton is uniquely determined by four parameters rather than
five.  We use the background density $\rho$ and background velocity
$(c,d)$ as three of these parameters along with either the normalized
soliton amplitude $0 < \nu = |\sin \phi| \le 1$ or the soliton angle
$0 < \beta \le \pi/2$, the two being related via \eqref{eq:1} through
\begin{equation}
  \label{eq:20}
  \nu^2 = 1 - \frac{(c \sin \beta - d \cos \beta)^2}{\rho} .
  % \cos \beta = \frac{ c \sqrt{M^2 + \nu^2 -
  %     1} + d \sqrt{1 - \nu^2}}{\sqrt{\rho} M^2} ,
\end{equation}

The spatial dark soliton exhibits either an absolute or convective
instability depending on the Mach number of the background flow
\eqref{eq:5} and either the amplitude $\nu$ or, equivalently, the
soliton angle $\beta$.  By moving in the reference frame $\xi = x -
\kappa t$ of the normalized dark soliton~\eqref{eq:17}, the background
flow has velocity $-\kappa$ normal to the soliton and velocity $-w$
parallel to the soliton.  The critical Mach number of the background
flow and its first order asymptotic approximation are
\begin{subequations}
  \label{eq:22}
  \begin{align}
    \label{eq:Mcr}
    \Mcr(\nu) &= \sqrt{\kappa^2 + \wcr(\nu)^2} = \sqrt{1 - \nu^2 +
      \wcr(\nu)^2} \\
    \label{eq:6}
    &\stackrel{\textrm{(\ref{eq:13})}}{=} 1 + \frac{2}{9} \nu^4 \ + \
    O(\nu^6) ~.
  \end{align}
\end{subequations}
Transverse perturbations are absolutely unstable for $M < \Mcr$ and
convectively unstable for $M \ge \Mcr$.  \fig{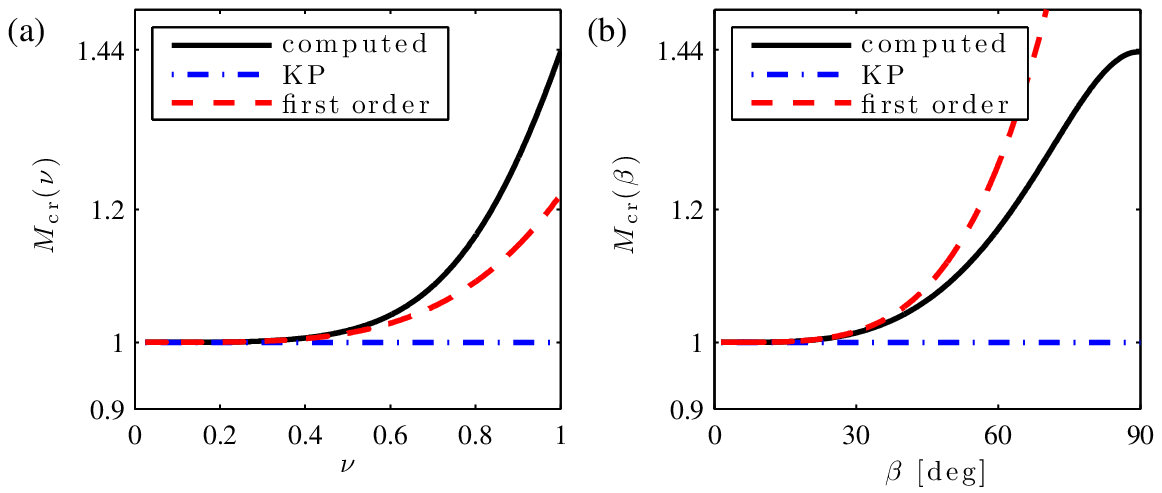}{Mcr}{1}{ Critical
  Mach number $\Mcr$ (Eq.~\eqref{eq:Mcr}) and its asymptotic
  approximations ($\Mcr = 1$ for KP, Eq.~(\ref{eq:6}) for the first
  order correction) as functions of the soliton amplitude $\nu$ (a)
  and the spatial soliton angle $\beta$ (b).  For $M \ge \Mcr$, the
  spatial dark soliton is convectively unstable, otherwise it is
  absolutely unstable. }

We also compute the critical Mach number's dependence on the soliton
angle by use of eq.~\eqref{eq:20} leading to a transformation between
$\nu$ and $\beta$
\begin{align*}
  \tan \beta &= \frac{w}{\sqrt{1 - \nu^2}} ~.
\end{align*}
Therefore, shallow spatial dark solitons at critical velocity $w =
\wcr(\nu)$ have a small $\beta = \nu + \mathcal{O}(\nu^3) \ll 1$, so
that
\begin{equation}
  \label{eq:21}
  \Mcr(\beta) = 1 + \frac{2}{9} \beta^4 + \mathcal{O}(\beta^6) ~.
\end{equation}

Figure \ref{fig:Mcr} shows the numerically calculated dependence of
$\Mcr$ on $\nu$ and $\beta$ and comparisons with the asymptotic
results \eqref{eq:6} and \eqref{eq:21}.  Combining the asymptotic
result (\ref{eq:6}) with these computations leads to
\begin{concl}
  \label{prop:transition}
  The transition between convective and absolute instability for
  spatial dark solitons always occurs at supersonic speeds $\Mcr >
  1$.  A sufficient condition for a spatial dark soliton with
  background Mach number $M$ to be \textbf{absolutely} unstable is
  \begin{equation*}
    M \le 1 ~.
  \end{equation*}
  A sufficient condition for a spatial dark soliton with background
  Mach number $M$ to be \textbf{convectively} unstable is
  $$
  M \ge \Mcr(\nu = 1) \approx 1.4374~.
  $$ 
  Additionally, $\nu \mapsto \Mcr(\nu)$ is monotonically increasing.
  In sum,
  \begin{align}
    \label{eq:106}
    1 < \Mcr \lessapprox 1.4374 ~.
  \end{align}
\end{concl}

\begin{remark}
  In \cite{kamchatnov-pitaevskii-08}, the bounds $1 \lesssim \Mcr
  \lesssim 1.46$ were obtained.  The leading order term in
  eq.~(\ref{eq:asympt-eig}) was used to show that $\Mcr \sim 1$ in the
  shallow regime.  Equation (\ref{eq:6}) improves the lower bound on
  $\Mcr$ and demonstrates that $\Mcr$ is strictly supersonic for all
  finite soliton amplitudes.  The upper bound $1.46$ in
  \cite{kamchatnov-pitaevskii-08} was calculated from a rational
  approximation of the spectrum for large soliton amplitudes
  \cite{pelinovsky_self-focusing_1995}.  Equation (\ref{eq:106}) gives
  the accurate upper bound.
\end{remark}

%%%%%%%%%%%%%%%%%%%%%%%%%%%%%%%%%%%%%%%%%%%%%%%%%%%%%%

\section{Oblique dispersive shock waves}
\label{sec:M-t-b}

In a dispersive fluid where dissipation is negligible, a jump in the
density/velocity may be resolved by an expanding oscillatory region
called a dispersive shock wave.  The Whitham averaging technique
\cite{whitham_non-linear_1965} has been successfully used to describe
a DSW's long time asymptotic behavior in a number of physical systems,
for example
\cite{gurevich_nonstationary_1974,gurevich_dissipationless_1987,hoefer_dispersive_2006,el_unsteady_2006,el_theory_2007}.
We briefly recap the rudiments of DSW theory.  A DSW is a modulated
wavetrain composed of a large amplitude, soliton edge and a small
amplitude, oscillatory edge, each moving with different speeds.  In
the relatively simple case where a DSW connects two constant states,
the speeds associated with each edge are determined by jump conditions
\cite{el_resolution_2005}, in analogy with the Rankine-Hugoniot jump
conditions of classical, viscous gas dynamics.  The jump conditions
result from a simple wave solution of the Whitham modulation equations
connecting the zero wavenumber, soliton edge to the zero amplitude,
oscillatory edge.  The existence of a DSW for a particular jump in the
fluid variables is guaranteed when an appropriate entropy condition is
satisfied.  For a left-going DSW, we define the leading (trailing)
edge to be the leftmost (rightmost) edge -- and vice versa for a
right-going DSW. The sign of the dispersion determines the locations
of the soliton and small amplitude edges.  For systems with positive
dispersion such as the NLS eq.~(\ref{eq:NLS}), the soliton is a
depression wave that resides at the trailing edge of the DSW.

While DSWs in (1+1)-dimensions have been well-studied, the theory of
supersonic dispersive fluid dynamics in multiple spatial dimensions is
in its infancy.  Perhaps the simplest DSW in multiple dimensions is an
oblique DSW, which has been studied in the stationary
\cite{gurevich_supersonic_1995,gurevich_supersonic_1996,el_spatial_2006,el_two-dimensional_2009}
and non-stationary \cite{hoefer_theory_2009} regimes (see
Figs.~\ref{fig:oblique_dsw} and \ref{fig:corner_dsw}).  In this
section, the analysis from the previous section is applied to the
stationary and nonstationary oblique DSW soliton trailing edge to
determine the separatrix between convective and absolute
instabilities.  In addition, in the weak shock and hypersonic regimes,
we find that the jump conditions for stationary and nonstationary
oblique DSWs are the same.  As in classical gas dynamics, oblique DSWs
can serve as building blocks for more complicated boundary value
problems.  Therefore, understanding the instability properties of
oblique DSWs is important and relevant to supersonic dispersive flows.
This has been further demonstrated by recent numerical simulations of
NLS supersonic flow past a corner
\cite{el_two-dimensional_2009,hoefer_theory_2009}.

In Sec.~\ref{sec:oblique-dsw-jump}, the jump conditions and
instability properties of nonstationary oblique DSWs are presented.
The following Sec.~\ref{sec:stat-weak-obliq} contains a derivation of
a stationary oblique DSW in the shallow regime, its stability, and
comparisons with numerical simulation.  Finally,
Sec.~\ref{sec:relat-betw-stat} demonstrates the connections between
stationary and nonstationary oblique DSWs.

\subsection{Nonstationary Oblique DSWs}
\label{sec:oblique-dsw-jump}

\fig{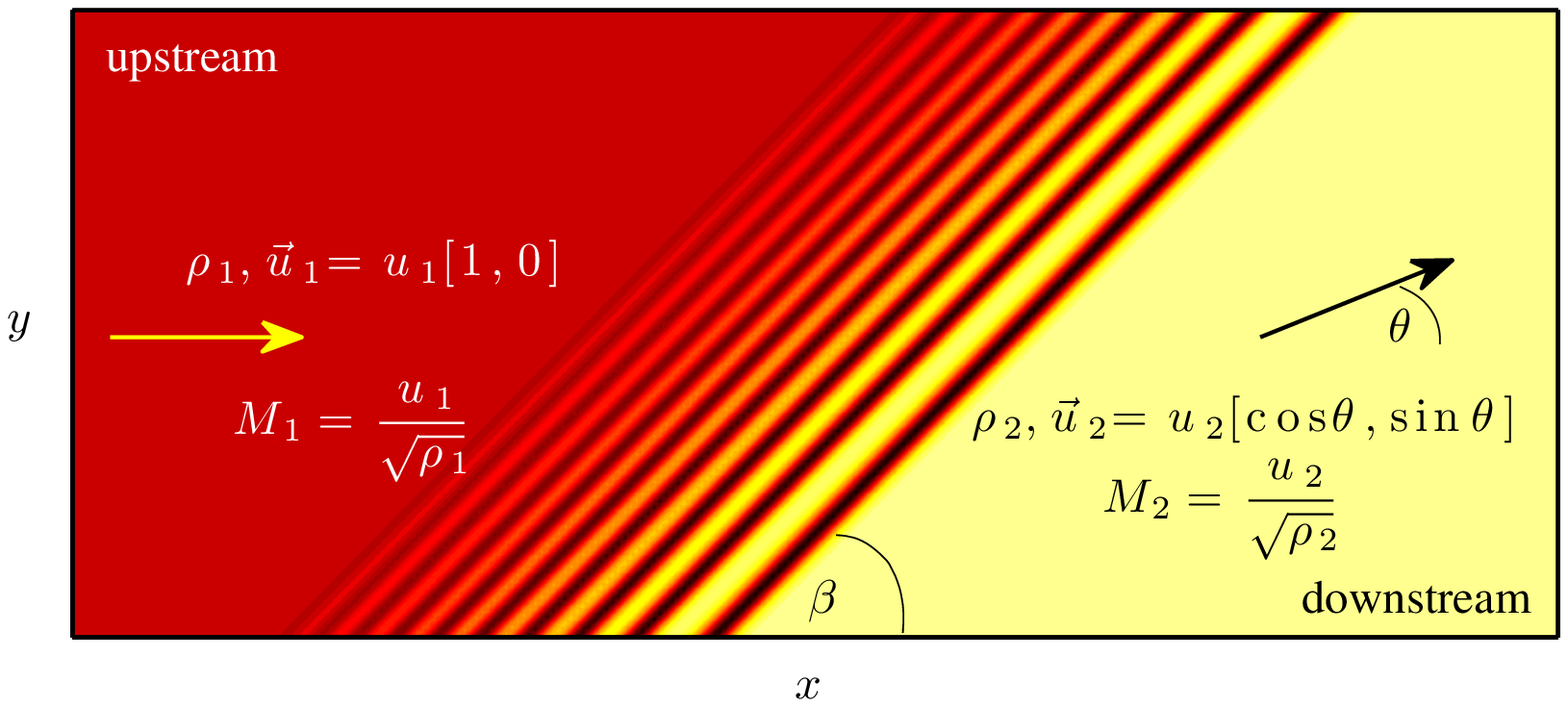}{oblique_dsw}{0.7}
{Schematic of an oblique DSW.}

In this section, we first recap the derivation of a nonstationary
oblique DSW \cite{hoefer_theory_2009} and then discuss its instability
properties.

A schematic of a non-stationary oblique DSW at a specific time in its
evolution is depicted in Fig.~\ref{fig:oblique_dsw}.  An incoming
upstream, supersonic flow is turned through the oblique DSW by the
deflection angle $\theta$.  To accommodate the deflection, the oblique
DSW expands along the wave angle $\beta$.  The leading edge consists
of small amplitude waves propagating into the upstream flow while the
trailing edge is composed of a dark soliton whose amplitude and speed
are asymptotically calculated from the oblique DSW jump conditions.

The nonstationary oblique DSW results from the long time evolution of
an initial jump in the density and velocity component normal to the
DSW wave angle $\beta$, in the direction $\hat{n}_\beta = (\sin \beta,
- \cos \beta)$, and continuity of the velocity parallel to $\beta$, in
the direction $\hat{p}_\beta = (\cos \beta, \sin \beta)$.  We consider
the upstream state
\begin{equation*}
  \label{eq:94}
  \lim_{x \to -\infty} \rho = \rho_1 ~, \quad \lim_{x \to -\infty}
  \vec{u} = (u_1,0)~,
\end{equation*}
and the downstream state
\begin{equation*}
  \label{eq:95}
  \lim_{x \to + \infty} \rho = \rho_2 ~, \quad \lim_{x \to + \infty}
  \vec{u} = (u_2 \cos \theta, u_2 \sin \theta)~.
\end{equation*}
The normal 1-DSW associated with the dispersionless characteristic
$\lambda_1 = u - \sqrt{\rho}$ (left-going wave) satisfies the simple
wave condition \cite{gurevich_dissipationless_1987}
\begin{equation}
  \label{eq:93}
  \hat{n}_\beta \cdot (u_1 - u_2 \cos \theta, - u_2 \sin \theta) = 2 (
  \sqrt{\rho_2} - \sqrt{\rho_1} ) ~.
\end{equation}
A NLS governed fluid experiences potential flow (see
eq.~(\ref{eq:Madelung})).  By restricting the spatial variation of the
solution to the direction $\hat{n}_\beta$ and integrating the
irrotationality constraint $v_x = u_y$ along the direction
$\hat{p}_\beta$, we obtain the continuity of the parallel velocity
component
\begin{equation}
  \label{eq:96}
  \hat{p}_\beta \cdot (u_1 - u_2 \cos \theta, - u_2 \sin
  \theta) = 0 ~ . 
\end{equation}
Choosing the reference frame in which the soliton trailing edge is
fixed, the speed of the soliton edge satisfies
\cite{gurevich_dissipationless_1987}
\begin{equation}
  \label{eq:97}
  \hat{n}_\beta \cdot (u_1,0) - \sqrt{\frac{\rho_2}{\rho_1}} = 0 ~.
\end{equation}
The jump conditions (\ref{eq:93}), (\ref{eq:96}), and (\ref{eq:97})
for the oblique DSW relate the upstream quantities $\rho_1$, $u_1$ and
one of the angles $\theta$ or $\beta$ to the downstream quantities
$\rho_2$, $u_2$ and the other angle.  Introducing the Mach numbers
$M_j = u_j/\sqrt{\rho_j}$, $j=1,2$ along with some manipulation, the
jump conditions become \cite{hoefer_theory_2009}
\begin{subequations}
  \label{eq:57}
  \begin{align}
    \label{eq:52}
    \tan(\beta - \theta) &= \frac{2}{M_1} \sec \beta -
    \tan \beta ~,
    \\
    \label{eq:53}
    M_2 &= \frac{\cot \beta}{\cos (\beta - \theta)} = \frac{\sqrt{M_1^2
        + 4 - 4 M_1 \sin \beta}}{M_1 \sin \beta} ~,
    \\
    \label{eq:54}
    \rho_2 &= \rho_1 M_1^2 \sin^2 \beta ~.
  \end{align}
\end{subequations}
Further manipulations lead to the equivalent relations
\begin{equation*}
  \frac{\sin (2 \beta - \theta)}{\cos (\beta - \theta)} =
  \frac{2}{M_1} ~, \quad \cos \theta = \frac{M_1 \cos (2\beta) + 2 \sin
    \beta}{\sqrt{4 + M_1^2 - 4 M_1 \sin \beta}} ~.
\end{equation*}

The associated entropy condition is $\rho_2 > \rho_1$, which, when
incorporated into the jump conditions, gives
\begin{equation*}
  M_1 > 1 ~, \quad 0 \le \theta \le \pi ~, \quad \sin^{-1}\frac{1}{M_1}
  \le \beta \le \frac{\pi}{2} ~.
\end{equation*} 
These state that the upstream flow must be supersonic, the flow always
turns \emph{into} the DSW, and the wave angle is larger than the Mach
angle $\sin^{-1} (1/M_1)$.  The Mach angle is half the opening angle
of the Mach cone inside of which infinitesimally small disturbances
are confined to propagate in dispersionless supersonic flow.  A
convenient way to visualize these results is by the
$M$-$\theta$-$\beta$ diagram in Fig.~\ref{fig:M_t_b_convective} that
relates the deflection and wave angles for a given upstream Mach
number $M_1$.  Figure \ref{fig:M_t_b_convective} includes the sonic
curve $M_2 = 1$ (to the right/left the flow is sub/supersonic).

\emph{A natural question is whether oblique DSWs with supersonic
  downstream flow conditions are convectively or absolutely unstable.}
To address this question, we use:
\begin{defin}
  \label{sec:nonst-obliq-dsws}
  Transverse perturbations to the nonstationary and stationary oblique
  DSW are convectively (absolutely) unstable whenever the trailing,
  dark soliton edge is convectively (absolutely) unstable.
\end{defin}

\fig{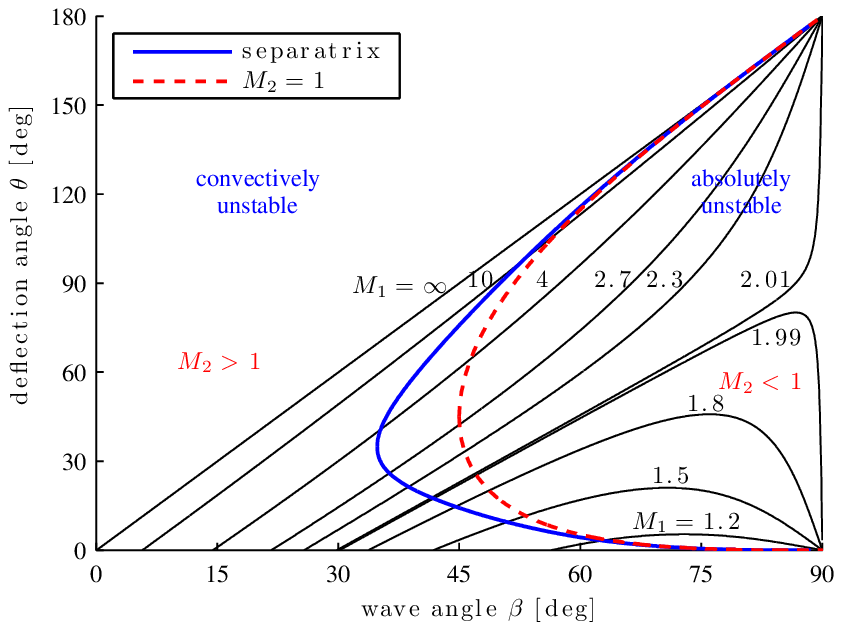}{M_t_b_convective}{0.83} { The
  $M$-$\theta$-$\beta$ diagram for non-stationary oblique DSWs of the
  NLS equation~(\ref{eq:NLS}).  Each upstream Mach number $M_1$ leads
  to a relationship between $\theta$ and $\beta$.  The separatrix
  curve (solid) between convectively and absolutely unstable solitons
  is supersonic, \ie in the $M_2>1$ region (left of the sonic line,
  dashes).  The separatrix curve asymptotes to the sonic line as
  $\beta \to 90^\circ$. }%\\[3mm]
%\noindent
%\textbf{Oblique DSW convective/absolute separatrix}

% \subsubsection{Oblique DSW convective/absolute separatrix}
% \label{sec:obliq-dsw-conv}
See further discussion in Sec.~\ref{sec:discuss}.

Spatial dark solitons exhibit the constraint \eqref{eq:1}.  When
applied to the oblique DSW trailing edge in Fig.~\ref{fig:oblique_dsw}
with background flow parameters $(c,d) = \sqrt{\rho} M_2 ( \cos \theta,
 \sin \theta)$, we find
\begin{equation*}
  \cos \phi = M_2 \sin (\beta - \theta) ~.
\end{equation*}
Using the jump conditions in eqs.~(\ref{eq:57}), we determine the
normalized soliton amplitude
\begin{equation}
  \label{eq:46}
  \nu(M_1,\theta) = \sin \phi = \frac{2 \sqrt{M_1 \sin \beta - 1}}{M_1
    \sin \beta} ~,
\end{equation}
where $\beta$ is related to $\theta$ by~\eqref{eq:52}.  The Mach
number of the downstream flow adjacent to the soliton is $M_2$ so the
absolute/convective stability criterion \eqref{eq:Mcr} determines the
separatrix
\begin{equation}
  \label{eq:12}
  M_2(M_1,\theta) = \Mcr(\nu) = \sqrt{1 - \nu^2 + \wcr(\nu)^2} ~ ,
\end{equation}
with $\nu$ given in~\eqref{eq:46}.  Conclusion \ref{prop:transition}
implies
\begin{cor}
  \label{cor:nonst-obliq-dsws-1}
  Nonstationary oblique DSWs with subsonic downstream flow are
  absolutely unstable.  Supersonic downstream flow can be either
  convectively or absolutely unstable.
\end{cor}

This conclusion can also be gleaned from
Fig.~\ref{fig:M_t_b_convective}.  To the right of the separatrix, the
trailing edge oblique soliton is absolutely unstable because $M_2 <
\Mcr$ while to its left, the soliton is convectively unstable.  The
region to the right of the separatrix and to the left of the sonic
line represents absolutely unstable oblique DSWs with supersonic
downstream flow conditions.  Below we derive additional properties of
the separatrix.

From Fig.~\ref{fig:M_t_b_convective}, we observe a minimum wave angle
$\bcr$, below which the oblique DSW is convectively unstable.  Setting
$M_2 = \Mcr$ in eq.~\eqref{eq:53} and solving for $\beta$ we find
\begin{equation*}
  \sin \bcr = \frac{-2 + \sqrt{4 + (4 + M_1^2) \Mcr^2}}{M_1 \Mcr^2}
  \, , 
\end{equation*}
which has a minimum for $M_1 = 2 \sqrt{1 + \Mcr^2}$, $\Mcr = \Mcr(\nu
= 1) \approx 1.4374$.  We therefore have a sufficient condition for
the oblique DSW trailing edge to be convectively unstable
\begin{equation*}
  \beta \le \bcr = \sin^{-1} \left [ (1 + \Mcr(1)^2)^{-1/2}\right ] \approx
  34.83^\circ .
\end{equation*}

The nonstationary oblique DSW is uniquely determined by the parameters
$M_1$, $\theta$, and $\rho_1$.  Thus, given $M_1 > 1$ and $0 < \theta
< \pi$, the absolute or convective instability of the corresponding
oblique DSW's trailing edge is determined by the location of
$(M_1,\theta)$ relative to the separatrix condition~\eqref{eq:12} in
the $M_1$-$\theta$ plane as shown in Fig.~\ref{fig:crit_sonic}.  The
parameter $\rho_1$ does not affect the absolute or convective nature
of the instability.  \fig{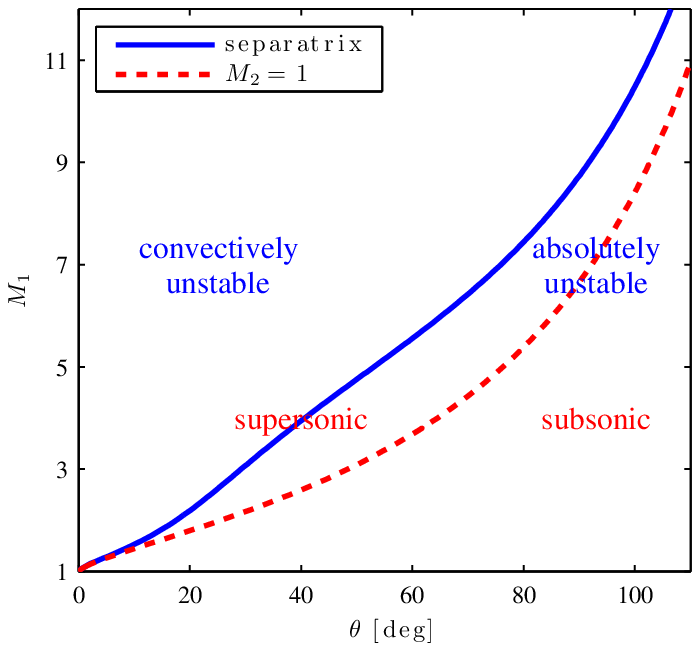}{crit_sonic}{0.67} { Separatrix
  curve (solid) dividing the $M_1$-$\theta$ plane for nonstationary
  oblique DSWs into two regions: the absolute/convective instability
  of the trailing edge dark soliton.  The sonic curve $M_2(M_1,\theta)
  = 1$ (dashed) is included for comparison.  }

% \subsubsection{Oblique DSW asymptotics}
% \label{sec:obliq-dsw-asympt}

Using the small amplitude result \eqref{eq:Mcr}, $\Mcr \sim 1 +
\frac{2}{9}\nu^4$, assuming near sonic upstream flow $M_1 = 1+ \eps$,
$0 < \eps = \mathcal{O}(\nu^2) \ll 1$, and expanding $\beta$,
$\theta$, $M_2$, and $\Mcr$, we compute the critical angle
\begin{align*}
  M_2(1+\eps,\thcr) &= \Mcr + \mathcal{O}(\eps^3) = \Mcr +
  \mathcal{O}(\nu^6) ~, \\ 
  \thcr &\sim \frac{4}{3 \sqrt{3}} \eps^{3/2} \left ( 1 -
    \frac{26}{27} \eps \right ) ~, \quad 0 < \eps \ll 1 ~.
\end{align*}
For $\theta \le \thcr$, the trailing edge dark soliton is convectively
unstable and absolutely unstable otherwise.  Similarly, the sonic
angle satisfies
\begin{align*}
  M_2(1+\eps,\thsonic) &= 1 + \mathcal{O}(\eps^3) = 1 +
  \mathcal{O}(\nu^6) ~, \\
  \thsonic &\sim \frac{4}{3 \sqrt{3}} \eps^{3/2} \left ( 1 -
    \frac{2}{3} \eps \right ) , \quad 0 < \eps \ll 1~.
\end{align*}
For $\theta < \thsonic$, the downstream flow is supersonic and
subsonic when $\theta > \thsonic$.  For the narrow window of
deflection angles $\thcr < \theta < \thsonic$, the flow is supersonic
and absolutely unstable.

\subsection{Spatial Oblique DSWs}
\label{sec:stat-weak-obliq}

We have so far focused on nonstationary oblique DSWs.  In this
section, we construct \emph{stationary} or spatial oblique DSWs in the
weakly nonlinear regime (see Fig.~\ref{fig:corner_dsw}), study their
instability properties, and perform numerical simulations.  This
discussion for the NLS equation (\ref{eq:NLS}) with positive
dispersion parallels the developments in
\cite{gurevich_supersonic_1995,gurevich_supersonic_1996} applied to
ion-acoustic waves in plasma, a system with negative dispersion.

Equations (\ref{eq:23}), (\ref{eq:24}), and the irrotationality
constraint due to potential flow are considered in the
(2+0)-dimensional case
\begin{subequations}
  \label{eq:75}
  \begin{align}
    \label{eq:30}
    (\rho u)_x + (\rho v)_y &= 0 ~, \\
    \label{eq:36}
    uu_x + v u_y + \rho_x &= \frac{1}{4} \left ( \frac{\rho_{xx} +
        \rho_{yy}}{\rho} - \frac{\rho_x^2 + \rho_y^2}{2 \rho^2} \right
    )_x ~, \\
    \label{eq:2}
    v_x - u_y &= 0 ~.
  \end{align}
\end{subequations}
We seek a special class of solutions that are related to supersonic
flow past a sharp corner or wedge.  For this, we treat $y$ as a
time-like variable and consider the ``initial conditions'' at $y = 0$
\begin{subequations}
  \label{eq:76}
  \begin{align}
    \label{eq:38}
    \rho(x,0) &= \left \{
      \begin{array}{cc}
        1, & x < 0 \\
        \rho_2, & x > 0
      \end{array} \right ., \\
    \label{eq:42}
    u(x,0) &= \left \{
      \begin{array}{cc}
        M_1, & x < 0 \\
        \sqrt{\rho_2} M_2 \cos \theta , & x >0
      \end{array} \right ., \\
    \label{eq:68}
    v(x,0) &= \left \{
      \begin{array}{cc}
        0, & x < 0 \\
        \sqrt{\rho_2} M_2 \sin \theta , & x > 0
      \end{array} \right . .
  \end{align}
\end{subequations}
The well-posedness of this initial value problem is plausible in the
supersonic regime, $M_j > 1$, $j = 1,2$, due to the hyperbolicity of
the dispersionless equations (see \eg \cite{courant_supersonic_1948}).
We seek a stationary, oblique DSW solution in the supersonic and
weakly nonlinear regime $0 < \rho_2 - 1 \ll 1$.  For this, we apply
the method of multiple scales
\begin{subequations}
  \label{eq:77}
  \begin{align}
    \label{eq:50}
    \rho &= 1 + \eps \rho^{(1)} + \eps^2 \rho^{(2)} + \cdots~, \\
    \label{eq:55}
    u &= M_1 - \eps u^{(1)} + \eps^2 u^{(2)} + \cdots ~, \\
    \label{eq:56}
    v &= \eps v^{(1)} + \eps^2 v^{(2)} + \cdots ~,
  \end{align}
\end{subequations}
in the transformed variables
\begin{equation}
  \label{eq:58}
  \xi = \eps^{1/2} [x - (M_1^2 - 1)^{1/2} y] ~, \quad \tau = \eps^{3/2}
  y ~.
\end{equation}
This particular choice is motivated by the line $\xi = const$ whose
angle with the $x$ axis is the Mach angle $\sin^{-1} (1/M_1)$ for
small amplitude wave propagation in the upstream flow.  Equating like
powers of $\eps$ leads to
\begin{equation*}
  \mathcal{O}(\eps^{\frac{3}{2}}): \quad
  \begin{array}{rl}
    -u^{(1)}_\xi + M_1 \rho^{(1)}_\xi - (M_1^2 - 1)^{\frac{1}{2}} v^{(1)}_\xi
    =& 0 ~, \\[2mm]
    -M_1 u^{(1)}_\xi + \rho^{(1)}_\xi =& 0 ~, \\[2mm]
    v^{(1)}_\xi - (M_1^2 - 1)^{\frac{1}{2}} u^{(1)}_\xi =& 0 ~ .
  \end{array}
\end{equation*}
The solution incorporating the initial conditions (\ref{eq:76}) is
\begin{equation}
  \label{eq:61}
  \rho^{(1)} = M_1 u^{(1)}, \quad v^{(1)} = (M_1^2 -
  1)^{\frac{1}{2}} u^{(1)} ~,
\end{equation}
with $u^{(1)}$ determined at the next order:
\begin{equation*}
  \mathcal{O}(\eps^{\frac{5}{2}}): \quad
  \begin{array}{rl}
    u^{(2)}_\xi + M_1\rho^{(2)}_\xi - (M_1^2 - 1)^{\frac{1}{2}}
    v^{(2)}_\xi - (\rho^{(1)} u^{(1)})_\xi \quad &\\ + v^{(1)}_\tau
    - (M_1^2 - 1)^{\frac{1}{2}} (\rho^{(1)} v^{(1)})_\xi = &0 ~, \\[2mm]
    M_1 u^{(2)}_\xi + \rho^{(2)}_\xi + u^{(1)} u^{(1)}_\xi + (M_1^2 -
    1)^{\frac{1}{2}} v^{(1)} u^{(1)}_\xi = &\frac{1}{4}
    \rho^{(1)}_{\xi\xi\xi} ~, \\[2mm]
    v^{(2)}_\xi + (M_1^2 - 1)^{\frac{1}{2}} u^{(2)}_\xi + u^{(1)}_\tau
    = &0 ~ .
  \end{array}
\end{equation*}
Inserting eqs.~(\ref{eq:61}) we obtain the KdV equation
\begin{equation}
  \label{eq:69}
  u^{(1)}_\tau - \frac{3 M_1^3}{2 (M_1^2 - 1)^{\frac{1}{2}}} u^{(1)}
  u^{(1)}_\xi + \frac{M_1^2}{8 (M_1^2 - 1)^{\frac{1}{2}}} u^{(1)}_{\xi
    \xi \xi}= 0 ~.
\end{equation}
It is convenient to consider the transformed variables $U$, $\zeta$ as
\begin{equation}
  \label{eq:70}
  U = - \frac{3 M_1^{\frac{7}{3}}}{(M_1^2 -
    1)^{\frac{1}{3}}} u^{(1)} + 1, \quad \xi =
  \frac{M_1^{\frac{2}{3}} (\zeta - \tau)}{2(M_1^2 -
    1)^{\frac{1}{6}} } ~.
\end{equation}
Then, eq.~(\ref{eq:69}) becomes the KdV equation with negative dispersion
\begin{equation*}
  U_\tau + U U_\zeta + U_{\zeta \zeta\zeta} = 0 ~.
\end{equation*}
The initial data in (\ref{eq:76}) maps to the Riemann problem
\begin{equation*}
  \label{eq:74}
  U(\zeta,0) = \left \{
    \begin{array}{cc}
      1 & \zeta < 0 \\
      0 & \zeta > 0
    \end{array} \right . ~.
\end{equation*}
This \emph{dispersive} Riemann problem was solved by Gurevich and
Pitaevski{\u i} in 1974 \cite{gurevich_nonstationary_1974}.  The
result is a DSW with the trailing edge, small amplitude wave speed
$c_{\mathrm{T}} = -1$ and leading edge, soliton speed $c_{\mathrm{L}}
= 2/3$.  The leading edge soliton amplitude is $2$ corresponding to
the KdV soliton speed/amplitude relation.  The oscillatory part of the
DSW, for $\tau$ sufficiently large, has the approximate form
\cite{gurevich_nonstationary_1974,hoefer_dispersive_2009}
\begin{equation}
  \label{eq:98}
  U(\zeta,\tau) \sim m(\zeta/\tau) - 1 + 2 \mathrm{dn}^2 \left (
    \frac{K[m(\zeta/\tau)]}{\pi} \phi(\zeta,\tau) ; m(\zeta/\tau)
  \right ) ~, \quad \tau \gg 1 ~,
\end{equation}
where $\mathrm{dn}$ is a Jacobi elliptic function and $K[m]$ is the
complete elliptic integral of the first kind.  The elliptic parameter
$m(\zeta/\tau)$ is the self-similar, simple wave solution to the
Whitham modulation equations given implicitly by
\begin{equation*}
  \label{eq:100}
  \frac{\zeta}{\tau} = \frac{1}{3}(1 + m) - \frac{2}{3} m \frac{(1-m)
    K[m]}{E[m] - (1-m)K[m]} ~,
\end{equation*}
where $E[m]$ is the complete elliptic integral of the second kind.
The phase is determined through
\begin{equation*}
  \label{eq:99}
  \phi(\zeta,\tau) = - \frac{\pi \tau}{\sqrt{6}}
  \int_{\zeta/\tau}^{2/3} \frac{ dz}{K[m(z)]} ~.
\end{equation*}

To obtain the NLS oblique DSW solution in its unscaled form, we use
the transformations (\ref{eq:70}), (\ref{eq:58}) along with the
substitutions (\ref{eq:61}) to match the asymptotic solution
(\ref{eq:77}) to the initial conditions (\ref{eq:76}).  The deflection
angle $\theta$ is related to the small parameter $\eps$ via
\begin{equation}
  \label{eq:60}
  \theta \sim \eps \frac{(M_1^2 - 1)^{\frac{5}{6}}}{3 M_1^{\frac{10}{3}}}
  \ll 1 ~ ,
\end{equation}
so that weak spatial DSWs correspond to a small DSW deflection angle.
Then the relationship between the downstream and upstream variables
takes the asymptotic form
\begin{subequations}
  \label{eq:64}
  \begin{align}
    \label{eq:80}
    \rho_2 &\sim 1 + \frac{M_1^2}{(M_1^2 -1)^\frac{1}{2}} \theta ~, \\
    \label{eq:81}
    M_2 &\sim M_1\left ( 1 -  \frac{M_1^2 + 2}{2(M_1^2 - 1)^\frac{1}{2}}
      \theta \right ) ~, \quad 0 < \theta \ll 1 ~.
  \end{align}
\end{subequations}
The KdV DSW speeds $c_{\mathrm{T}} = -1$ and $c_{\mathrm{L}} = 2/3$
correspond to the slopes of the oscillatory region's boundaries which
we transform to the leading and trailing angles $\beta_+$, $\beta_-$,
respectively, for the stationary oblique DSW.  Using the
transformations (\ref{eq:58}), (\ref{eq:70}), and (\ref{eq:60}), the
oblique DSW angles take the asymptotic forms
\begin{subequations}
  \label{eq:73}
  \begin{align}
    \label{eq:78}
    \beta_- &\sim \sin^{-1}\left ( \frac{1}{M_1} \right ) +
    \frac{M_1^2}{2(M_1^2 - 1)} \theta ~, \\
    \label{eq:79}
    \beta_+ &\sim \sin^{-1}\left ( \frac{1}{M_1} \right ) +
    \frac{3M_1^2}{M_1^2 - 1} \theta ~, \quad 0 < \theta \ll 1 ~.
  \end{align}
\end{subequations}
Finally, the trailing edge soliton amplitude and phase jump $2 \phi$
with the angle $\beta^-$ have the asymptotic form
\begin{equation}
  \label{eq:86}
  \rho_2 - \rho(x,x \tan \beta^-) = \sqrt{\rho_2} \sin \phi \sim \phi
  \sim \frac{2 M_1^2}{(M_1^2 - 1)^\frac{1}{2}} \theta ~ . 
\end{equation}
This DSW solution is plotted in Fig.~\ref{fig:corner_dsw} and
approximates a stationary, weak, oblique DSW for NLS.
\fig{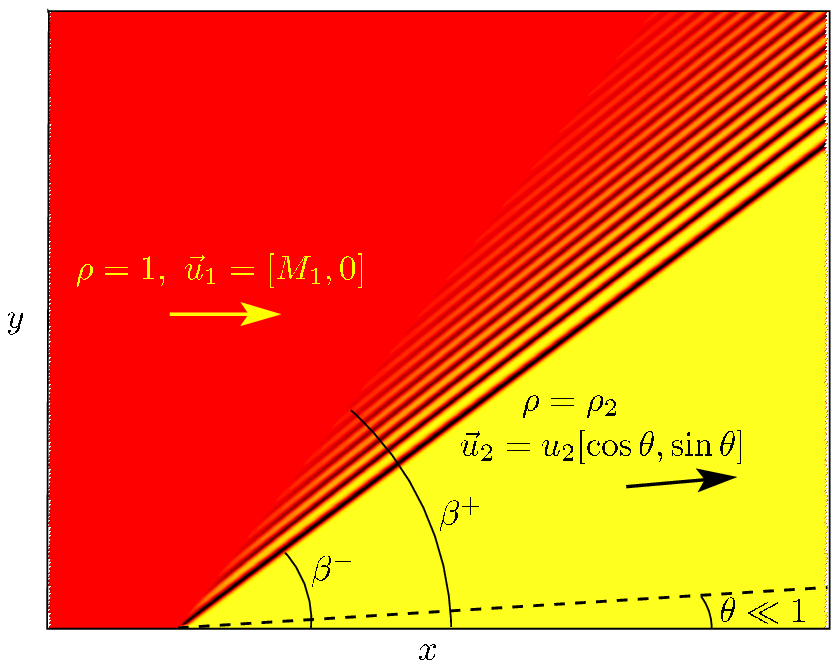}{corner_dsw}{0.5}{Example spatial
  oblique DSW with small deflection angle $\theta$ constructed from
  the asymptotic solution (\ref{eq:98}).}

Equations (\ref{eq:64}) and (\ref{eq:73}) are the jump conditions for
weak, stationary oblique DSWs.  These can be used to approximately
solve the problem of supersonic flow over a corner with angle $0 <
\theta \ll 1$.  Additionally, due to symmetry arguments, two
stationary oblique DSWs approximately solve supersonic flow over a
wedge as in \cite{el_two-dimensional_2009}.  Figure
\ref{fig:corner_sim} shows the numerical solution of
eq.~\eqref{eq:NLS} for supersonic $M_1 = 2$ flow past a corner with
angle $\theta = 9^\circ$ after the flow pattern has reached a
quasi-steady state (see Sec.~\ref{sec:numer-solut-nls} for the
numerical details).  Sufficiently close to the corner, the structure
of the numerical solution resembles the asymptotic oblique DSW shown
in Fig.~\ref{fig:corner_dsw}.  Further away from the corner, the first
sign of instability occurs along the trailing, dark soliton edge
leading to the generation of vortices.  This provides some
justification for our definition of oblique DSW instability in
Def.~\ref{sec:nonst-obliq-dsws}.  Furthermore, we observe that the
vortices are convected further away from the corner as time
progresses\footnote{The vortex pattern eventually stabilizes at a
  fixed distance from the corner.  A recent
  study~\cite{kamchatnov_condition_2011} of NLS dark soliton
  convective/absolute instabilities has some independent results that
  overlap with ours in Sections~\ref{sec:crit-separ-solit}
  and~\ref{sec:conv-inst-spat}.  This work also gives a further
  description of perturbation convection along the soliton.  The
  effective group velocity of the perturbation along the soliton is
  found to be equal to the critical flow speed (here
  $\sqrt{\rho_2}\wcr$).  However, convective instability theory does
  not explain the numerically observed stabilization of vortex
  formation at a fixed distance from the corner.}.  In a previous work
\cite{hoefer_theory_2009}, the authors performed numerical simulations
of NLS supersonic flow past a corner for a large number of flow
configurations, observing similar, stable pattern formation in some
cases.  Flow configurations where the instability overwhelms any
stable pattern formation were also observed.  We identify these two
flow regimes with convective and absolute instability of the oblique
DSW.  \fig{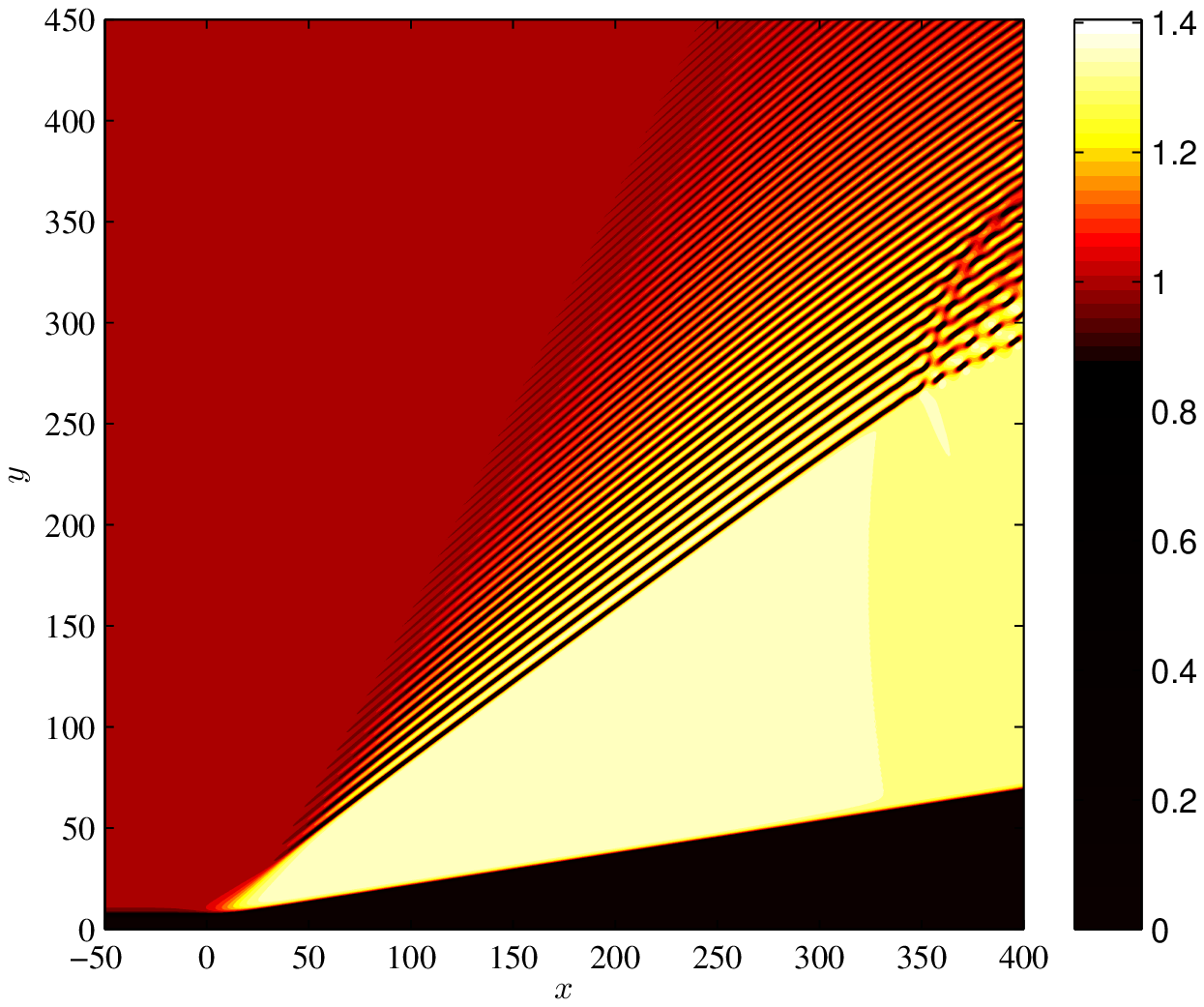}{corner_sim}{0.7}{Numerical simulation of
  supersonic flow past a corner with $\theta = 9^\circ$, $M_1 = 2$, at
  $t = 400$.  The color scale is chosen to visually resolve the small
  amplitude oscillations.}

Table \ref{tab:corner_compare} summarizes the asymptotic estimates in
eqs.~\eqref{eq:64} and \eqref{eq:73} compared with the numerical
computations showing excellent agreement, even for fairly large corner
angles and when the ``small'' parameter $\eps$ is larger than one.
\begin{table}
  \centering
  \begin{tabular}{|c|c|c|c|c|c|c|}
    \hline
    $M_1 = 2$ & $\theta$ & $\eps$ & $M_2$ & $\rho_2$ & $\beta^-$ & $\beta^+$ \\
    \hline\hline
    theory  & $3^\circ$ & 0.6 & 1.82 & 1.12 & $32^\circ$ & $42^\circ$ \\
    \hline
    numerics  & $3^\circ$ & 0.6 & 1.84 & 1.12 & $32^\circ$ & $39^\circ$ \\
    \hline \hline
    theory  & $6^\circ$ & 1.3 & 1.64 & 1.24 & $34^\circ$ & $54^\circ$ \\
    \hline
    numerics  & $6^\circ$ & 1.3 & 1.67 & 1.26 & $34^\circ$ &
    $49^\circ$ \\
    \hline\hline
    theory  & $9^\circ$ & 1.9 & 1.46  & 1.36 & $36^\circ$ & $66^\circ$ \\
    \hline
    numerics  & $9^\circ$ & 1.9 & 1.51 & 1.40 & $37^\circ$ &
    $62^\circ$ \\
    \hline
  \end{tabular}
  \caption{Comparison between the asymptotic results of
    eqs.~\eqref{eq:64} and \eqref{eq:73} and numerical simulation of
    supersonic flow over a corner.} 
  \label{tab:corner_compare}
\end{table}

The trailing edge dark soliton is shallow.  Therefore, using the
theory developed in Sec.~\ref{sec:asympt-separ-crit} and
eqs.~\eqref{eq:81}, \eqref{eq:6}, the oblique DSW is convectively
unstable when
\begin{equation*}
  \label{eq:87}
  M_2 > \Mcr(\nu) \quad \mathrm{or} \quad M_1[1 + \mathcal{O}(\theta)]
  > 1 + \mathcal{O}(\theta^4) ~, \quad 0 < \theta \ll 1 ~,
\end{equation*}
because $\nu = \sin \phi \sim \mathcal{O}(\theta)$ from
eq.~\eqref{eq:86}.  As long as $M_1 > 1$, independent of the corner
angle $\theta$, using Conclusion \ref{prop:transition} gives
\begin{corollary}
  \label{cor:spatial-oblique-dsws-1}
  For NLS supersonic upstream flow $M_1 > 1$ and sufficiently small
  corner angles $0 < \theta \ll 1$, the oblique DSW emanating from a
  sharp corner is convectively unstable.
\end{corollary}

\subsection{Relationship between stationary and nonstationary oblique DSWs}
\label{sec:relat-betw-stat}

As shown in the previous section, stationary oblique DSWs can be
physically realized as the solution of a two-dimensional
\emph{boundary value problem} involving supersonic flow.  In contrast,
the nonstationary oblique DSW studied in
Sec.~\ref{sec:oblique-dsw-jump} results from the solution of an
\emph{initial value problem}.  As we now demonstrate, the downstream
flow conditions for the stationary and nonstationary oblique DSW are
the same in two asymptotic regimes: weak shocks and hypersonic flow.

The downstream flow conditions and the stationary trailing edge
soliton in both the stationary and nonstationary oblique DSWs are
characterized by the deflection angle $\theta$, the wave angle
$\beta^-$ or $\beta$ for the nonstationary case, the Mach number
$M_2$, and the density $\rho_2$.  These properties are related via the
oblique DSW jump conditions.  For weak oblique DSWs, we assume a fixed
upstream supersonic Mach number $M_1 > 1$ and small deflection angle
$0 < \theta \ll 1$ as in Sec.~\ref{sec:stat-weak-obliq}.  By a
standard asymptotic calculation, an expansion of the jump conditions
for the nonstationary oblique DSW in eqs.~(\ref{eq:57}) in the form
\begin{equation*}
  \label{eq:102}
  \rho_1 = 1~, \quad \rho_2 = 1 +
  \mathcal{O}(\theta) ~, \quad M_2 = M_1 +
  \mathcal{O}(\theta) ~, \quad \beta = \sin^{-1} 1/M_1 +
  \mathcal{O}(\theta) ~,
\end{equation*}
gives precisely the same result as that obtained for the stationary
oblique DSW in eqs.~(\ref{eq:80}), (\ref{eq:81}), and (\ref{eq:78}).

The hypersonic regime assumes the large Mach number scaling $M_1 \gg
1$ and small deflection angle $\theta = \mathcal{O}(1/M_1)$.  In this
asymptotic regime, the jump conditions (\ref{eq:57}) become
\begin{subequations}
  \label{eq:90}
  \begin{align}
    \label{eq:10}
    \rho_2 &= \left ( \frac{\theta M_1}{2} + 1 \right )^2 +
    \mathcal{O}(1/M_1)~, \\
    \label{eq:89}
    M_2 &= \frac{2 M_1}{2 + \theta M_1} + \mathcal{O}(1) ~, \\
    \label{eq:91}
    \beta &= \frac{\theta}{2} + \frac{1}{M_1} + \mathcal{O}(1/M_1^2)
    ~, \quad M_1 \gg 1, \quad 0 < \theta = \mathcal{O}(1/M_1) ~,
  \end{align}
\end{subequations}
where we have assumed that $\rho_1 = 1$.  In
\cite{el_spatial_2006,el_two-dimensional_2009}, stationary oblique DSW
solutions of eqs.~\eqref{eq:75} and \eqref{eq:76} in the hypersonic
regime were constructed asymptotically.  The classical notion of
hypersonic similitude \cite{hayes_hypersonic_2004} applies so that the
(2+0)-dimensional stationary problem was asymptotically mapped to a
(1+1)-dimensional problem for the NLS equation.  Stationary,
supersonic flow past an extended obstacle is then related to a piston
problem, which can be solved analytically in the case of a sharp
corner (constant piston speed) \cite{hoefer_piston_2008} and for more
general profiles \cite{el_two-dimensional_2009,kamchatnov_flow_2010}.
The results for the stationary oblique DSW are the same as those
computed asymptotically for the nonstationary case in
eqs.~(\ref{eq:90}) when $M_1 \theta \le 2$.  The case $M_1 \theta > 2$
corresponds to a novel feature of the dispersive piston problem where
the oblique DSW experiences cavitation and the DSW forms an
oscillatory wake \cite{hoefer_piston_2008}, not captured by the jump
conditions (\ref{eq:57}).  Combining these results with Corollaries
\ref{cor:nonst-obliq-dsws-1} and \ref{cor:spatial-oblique-dsws-1}
demonstrates
\begin{concl}
  \label{prop:spatial-oblique-dsws}
  For weak (small deflection angle $0 < \theta \ll 1$, fixed upstream
  Mach number $M_1$) or hypersonic ($M_1 \gg 1$, $\theta =
  \mathcal{O}(1/M_1)$, $M_1 \theta \le 2$) oblique DSWs, the
  nonstationary and stationary flows have the same asymptotic
  downstream flow properties and trailing edge soliton
  amplitudes/angles.  In these regimes, the oblique DSWs are
  convectively unstable.
\end{concl}

%%%%%%%%%%%%%%%%%%%%%%%%%%%%%%%%%%%%%%%%%%%%%%%%%%%%%%

\section{Computational techniques}
\label{sec:numerical-methods}

In this section, we present details of our numerical methods for
computing the spectrum of transversely unstable perturbations as well
as derivatives of the dispersion relation via adjoint methods
(Sec.~\ref{sec:variational}).  Direct numerical simulations of NLS
supersonic flow over a corner are explained in
Sec.~\ref{sec:numer-solut-nls}.

\subsection{Computing the spectrum and derivatives of the dispersion relation}
\label{sec:variational}

Accurately computing the spectrum of the linearized NLS
equation~\eqref{eq:NLS}, finding the maximal growth
wavenumber~\eqref{eq:max-growth} and the critical wavenumber
\eqref{eq:inflection1} require a fine grid and sufficiently large
computational domain.  This turns out to be challenging, especially in
the small amplitude regime.  To achieve this, we employ a combination
of computational and analytical techniques explained below.

\begin{itemize}
\item The linearized operator in~\eqref{eq:eigen-problem} is realized
  using the centered, fourth order (sparse) finite difference stencil
  in $\xi$ for the Laplacian and other derivative operators.  Zero
  Dirichlet boundary conditions are embedded into the associated
  matrix.  We find that a domain size of $11/\nu$ serves well
  (increasing the domain size has negligible effect on the results).
\item For accuracy, the number of grid points along the transverse
  direction $\xi$ should scale as $1/\nu$. As $\nu$ decreases from 1
  to 0.01, we use $2^9-2^{13}$ grid points.  Using fewer points can
  lead to completely wrong results, either because $\kmax \to 0$ or
  because $\kcr \to \kcut+\ $ as $\nu\to 0$.
\item The discrete eigenvalue $\Omega_0(k)$ and its associated
  localized eigenfunction ${\bf f}_0(\xi;k) $ are computed using
  Matlab's sparse eigenvalue solver (\texttt{'eigs'} with
  \texttt{'SM'}).
\item One approach is to compute $\Omega_0(k)$ on a grid of $k$ values
  (as for Fig.~\ref{fig:Omega_plot}).  Then, $\Omega_0'(k)$
  (resp.~$\Omega_0''(k)$) can be computed using finite differences and
  minimized on the $k$ grid to find $\kmax$ (resp. $\kcr$).  This
  method turns out to be computationally expensive.  To overcome these
  challenges, an accurate and fast method is explained below.
\end{itemize}
 
Recall the eigenvalue problem~\eqref{eq:eigen-problem}.  As discussed
previously the discrete spectrum of $JL$ consists of two simple
eigenvalues of opposite signs, $\pm \Omega_0(k)$ (we choose the
positive sign), with the associated eigenfunction ${\bf f}_0(\xi;k)$.
Our main goal is to compute $\kmax$ such that $\Omega_0'(\kmax)=0$,
and $\kcr$ such that $\Omega_0''(\kcr)=0$.  This is achieved using the
following algorithm:
\begin{enumerate}
\item
  Compute the discrete spectrum at some (initial) $k$, \ie
  $\Omega_0(k)$ and ${\bf f}_0(\xi;k)$.
\item Apply adjoint methods to find exact expressions for
  $\Omega_0'(k)$ and $\Omega_0''(k)$, \ie
  Eqs.\ \eqref{eq:Omega_p}--\eqref{eq:Omega_pp} below.
\item Repeat steps 1--2 using a root finder to converge to $\kmax$ and
  $\kcr$.
\item Compute $\gmax = \Im\{\Omega_0(\kmax)\}$ 
  and / or $\Ocr=\Omega_0(\kcr)$ and $\wcr = \Omega'_0(\kcr)$.
\end{enumerate}

We proceed to derive the relevant expressions.  Use will be made of
the standard Pauli matrices,
\begin{equation}
  \label{eq:Pauli}
  \sigma_1 \doteq \begin{bmatrix} 0 & 1 \\ 1 & 0 \end{bmatrix}~, \quad 
  \sigma_2 \doteq  \begin{bmatrix} 0 & -i \\ i & 0 \end{bmatrix}~, \quad 
  \sigma_3 \doteq  \begin{bmatrix} 1 & 0 \\ 0 & -1 \end{bmatrix} ~,
\end{equation}
the reflection operator $R$, $R g(x) = g(-x)$, and the adjoint of an
operator will be denoted by $(\cdot)^\dagger$.

Differentiating~\eqref{eq:eigen-problem} with respect to $k$ gives
\begin{equation}
  \label{eq:L_p}
  \left ( \sigma_2 L_0 + \frac{1}{2} k^2 \sigma_2 + \Omega_0 \right ){\bf f}_0'
  = - (k \sigma_2 + \Omega_0') {\bf f}_0~,
\end{equation}
where $(\cdot)^{'}$ denotes differentiation with respect to $k$.
Solvability requires that \mbox{$(k \sigma_2 + \Omega_0') {\bf f}_0$}
be orthogonal to the nullspace of the adjoint operator to the
left-hand side of~\eqref{eq:L_p}.  In Appendix \ref{ap:adjoint} we
prove that this nullspace can be characterized as follows.
\begin{lemma}
  \label{lem:adjoint}
  For $k\in \C \setminus \left \{0,\pm \kcut \right \}$ and
  $(\Omega_0(k),\fzero(\xi;k))$ an eigenpair satisfying
  \begin{equation}
    \label{eq:49}
    \left (\sigma_2 L_0 +
    \frac{1}{2} k^2 \sigma_2 + \Omega_0 \right ){\bf f}_0 = 0 ~,
  \end{equation}
  we have
  \begin{equation}
    \label{eq:nullspace}
    {\rm ker} \left \{    \left ( \sigma_2 L_0 + 
        \frac{1}{2} k^2 \sigma_2 + \Omega_0 \right )^{\dagger} \right
    \} \ = \  {\rm span} \{ R \sigma_1 \fzero^* \} ~.
  \end{equation}
\end{lemma}

Using Lemma~\ref{lem:adjoint} and taking the inner product
of~\eqref{eq:L_p} with $\fzeroad$, the solvability condition reads
\begin{equation}
  \label{eq:Omega_p}
  \Omega_0'(k) = -k \varprod{\sigma_2 \fzero} ~ .
\end{equation}
Differentiating~\eqref{eq:L_p} with respect to $k$ gives
$$
  \left ( \sigma_2 L_0 + \frac{1}{2} k^2 \sigma_2 + \Omega_0 \right ){\bf f}_0''
  = -  (k \sigma_2 + \Omega_0') \fzero' -  (\sigma_2 + \Omega_0'')
  \fzero ~.
$$
Using the solvability condition we conclude  that
\begin{equation}
  \label{eq:Omega_pp}
  \Omega_0''(k) = -2 \varprod{(k \sigma_2 + \Omega_0')\fzero'} -
  \varprod{\sigma_2 \fzero}~.
\end{equation}
In summary, we compute $\Omega_0'(k)$ and $\Omega_0''(k)$ using
Eqs.~(\ref{eq:L_p}), (\ref{eq:49}), \eqref{eq:Omega_p}, and
\eqref{eq:Omega_pp}.  These computations are accurate and fast.  The
most time consuming operation is the computation of the discrete
spectrum of $L\ $.

\subsection{Numerical solution of the NLS equation}
\label{sec:numer-solut-nls}

In Section \ref{sec:stat-weak-obliq} we presented the numerical
solution of supersonic NLS flow past a corner.  The technique used was
the same as that presented in \cite{hoefer_theory_2009}.  We introduce
a linear potential with large contrast that acts as a penalization to
flow outside the domain.  Such volume penalization methods are
well-known in classical fluid dynamics (see, e.g.,
\cite{keetels_fourier_2007}).  In the context of BEC and optics,
superfluid flow around obstacles or boundaries are realized in
practice using electromagnetic waves or a variable refractive index,
both modeled as a spatially varying, linear potential.  The benefits
of this numerical technique include the use of a regular, Cartesian
mesh and highly accurate pseudospectral derivative calculations.

The time-dependent NLS / GP equation (\ref{eq:NLS}) with a linear
potential
\begin{equation}
  \label{eq:92}
  i \psi_t = -\frac{1}{2} (\psi_{xx} + \psi_{yy} ) + V(x,y,t) \psi + |\psi|^2
  \psi ~,
\end{equation}
was solved numerically using a pseudospectral, Fourier spatial
discretization and a fourth order Runge-Kutta explicit time stepper.
These computations were performed on a rectangular mesh of $N_x N_y$
equispaced grid points within the domain $[-L_x,L_x] \times
[-L_y,L_y]$.  Our choice of the potential
\begin{subequations}
  \begin{align}
    \nonumber
    V(x,y,t) &= \\
    \label{eq:103}
    &\quad ~ V_0 \Big [ 1 - H_\mu(L_x - |x| - \delta)  H_\mu(L_y - |y| -
    \delta) \\
    \label{eq:82}
    & \quad \qquad H_\mu(y - C(x-M_1 t)) \Big ] ~, \\
    \label{eq:83}
    C(\xi) &= - \tan(\theta)  \Big [ H_\mu(x_1 - \xi) - H_\mu(l)
    \Big ] -L_y + \delta ~, \\
    \label{eq:84}
    H_\mu(\xi) &= \frac{1}{2} + \frac{1}{2} \tanh(\xi/\mu) ~,
  \end{align}
\end{subequations}
models the boundary conditions corresponding to flow over a corner and
also serves to ``localize'' the solution so that a pseudospectral,
Fourier discretization with periodic boundary conditions can be
employed.  An example potential is shown in
Fig.~\ref{fig:corner_potential}.  The time-dependent potential
corresponds to a moving ramp.  The function $H_\mu$ is a regularized
Heaviside step function with transition width $\mu$.  The terms on
line \eqref{eq:103} effect the localization of $\psi$ to within
$\delta$ of the domain boundaries.  The terms on lines \eqref{eq:82}
and \eqref{eq:83} correspond to a moving ramp with corner angle
$\theta$ and apex located at $(x_0-M_1 t,-L_y+\delta)$.  When the
second corner at $x = x_0 - M_1 t + l$ is reached, the ramp flattens
and continues as a straight line.  The initial condition is the
nonlinear, stationary ground state of eq.~(\ref{eq:92}) with potential
$V(x,y,0)$ computed by the spectral renormalization technique
\cite{ablowitz_spectral_2005} with the unit density constraint
$|\psi(0,0,0)|^2 = 1$.  The potential contrast $V_0$ is taken
sufficiently large so that the density is effectively zero where
$V(x,y,t) \approx V_0$.  Time integration was carried out until the
corner reached the left boundary.  Near the corner, the flow
approximates a ``pure'' oblique DSW as shown in
Fig.~\ref{fig:corner_sim}.  Parameter values for table
\ref{tab:corner_compare} are $N_x = 4000$, $N_y = 1000$, $L_x = 2000$,
$L_y = 500$, $V_0 = 20$, $\delta = 2$, $x_0 = 1000$, $l = 1000$, $\mu
= 2$, and a time step of 0.05.  The simulation depicted in
Fig.~\ref{fig:corner_sim} results from $N_x = 3200$, $N_y = 1600$,
$L_x = 800$, $L_y = 400$, $V_0 = 20$, $\delta = 2$, $x_0 = 400$, $l =
400$, $\mu = 2$, and a time step of 0.01.

\fig{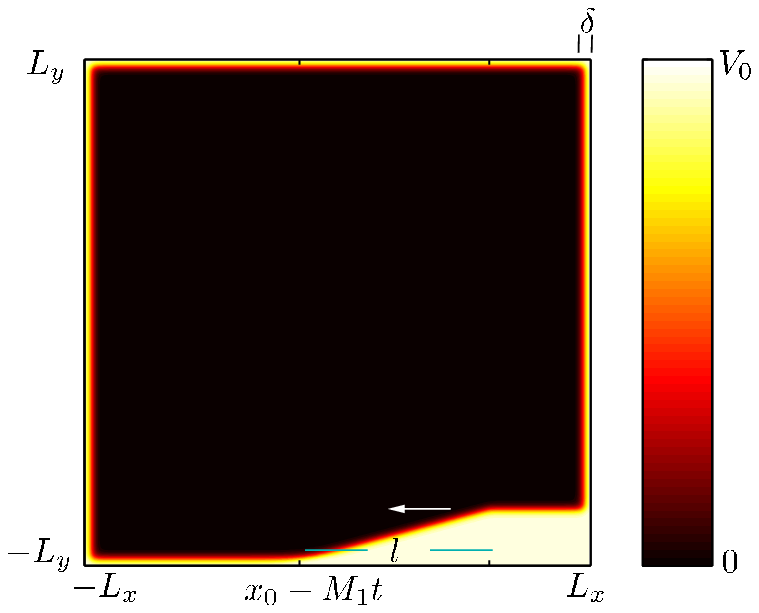}{corner_potential}{0.5}{An example
  potential $V(x,y,t)$ with $V_0 \gg 1$ used to model numerical
  simulation of supersonic flow past a corner.  Regions where
  $V(x,y,t)$ is large correspond to negligible density.  The ramp
  moves to the left with speed $M_1$ leading to oblique DSW
  formation.}

%%%%%%%%%%%%%%%%%%%%%%%%%%%%%%%%%%%%%%%%%%%%%%%%%%%%%%

\section{Discussion and conclusions}
\label{sec:discuss}

One of the motivating questions for this study was the nature of
convective versus absolute instabilities of dark solitons.  In
general, the characterization of the instability type requires
knowledge of the dispersion relation for a range of wavenumbers in the
complex plane.  Unfortunately, the exact discrete spectrum (and hence
dispersion relation) for NLS dark solitons is unknown.  The formal
analysis presented in~\cite{kamchatnov-pitaevskii-08} led to greatly
simplified criteria for determining the instability type, which
involve only the imaginary (stable) portion of the spectrum.  

In this study, the underlying assumptions behind the simplified
criteria are exposed and justified using a combination of rigorous
results (Theorem~\ref{theo:spectrum} and Lemma~\ref{lem:adjoint}),
shallow amplitude asymptotics (Prop.~\ref{prop:asympt-eig}), and
computations of the spectrum.  Consequences of the small-amplitude
asymptotics and numerical computations are the first order corrections
to the maximal growth rate and associated wavenumber
(Corollary~\ref{cor:max-growth}) and dependence of the critical Mach
number on the soliton amplitude (Conclusion \ref{prop:transition}).
Applying Conclusion \ref{prop:transition} to the soliton trailing edge
of oblique DSWs, we conclude that {\bf subsonic oblique DSWs are
  always absolutely unstable, whereas supersonic oblique DSWs can be
  absolutely or convectively unstable} (Corollaries
\ref{cor:nonst-obliq-dsws-1} and \ref{cor:spatial-oblique-dsws-1}).
In addition, the relationship between stationary DSWs (corner BVPs)
and nonstationary DSWs (Riemann IVPs) is studied.  In both cases, the
DSWs are found to have the same downstream flow properties in the
shallow and hypersonic regimes (Conclusion
\ref{prop:spatial-oblique-dsws}).

It is worth contrasting these results with oblique shock waves in
classical gas dynamics.  Supersonic classical shock fronts in gas
dynamics are linearly \textbf{stable} when they satisfy the Lax
entropy condition \cite{majda_stability_1983}.  For the boundary value
problem of supersonic flow past a sharp corner, the oblique shock is
stable if and only if the downstream flow is supersonic
\cite{li_analysis_2004,tkachev_courant-friedrichs_2009}.  As far as we
know, the distinction between absolute and convective instabilities in
the subsonic case has not been elucidated.  We note that recent
experiments in another viscous medium (granular material) exhibit the
stable excitation of oblique DSWs with both supersonic and subsonic
downstream flow conditions \cite{gray_weak_2007}.

Several questions and open problems related to this study are
mentioned below.

The nonstationary oblique DSW consists of a slowly modulated elliptic
function solution to NLS.  How to study the stability or instability
of this coherent structure is not immediately obvious given its
expanding nature and asymptotic representation.  The notion of
instability we consider here is centered upon the properties of
perturbations to the stationary, trailing dark soliton edge.  This is
a natural criterion because the soliton trailing edge corresponds to
the largest oscillation in the DSW, hence nonlinear effects are
strongest there.  Another motivation for this choice comes from the
numerical simulation of supersonic flow over a corner where the
instability first appears along the trailing edge soliton.  However,
to gain a more complete understanding of DSW instabilities, one should
develop an analysis of absolute and convective transverse
instabilities of elliptic function solutions.  This suggests the more
general study of convective/absolute instability for systems with
continuous bands of unstable modes.  We are not aware of any previous
work in this direction.

Careful computations of the spectrum suggest that
Conjecture~\ref{conj:eigs} is true. However, to the best of our
knowledge, it has not been proven rigorously.  It may be possible to
do this by reducing the problem to an ODE, where Sturm-Liouville
theory is applicable
(cf.~\cite{pego_eigenvalues_1992,rousset_simple_2010}).

It would be interesting to extend these results to systems with
negative dispersion such as shallow water waves where the KP-II
equation is valid in the small amplitude regime.  In contrast to KP-I
studied here, line solitons are linearly stable
\cite{alexander_transverse_1997}.  Are oblique DSWs in systems with
negative dispersion stable?

\textbf{Acknowledgments.}  The authors thank Mark Ablowitz for
inspiring remarks. The authors also thank Anatoly Kamchatnov for
constructive discussions and sharing his recent
manuscript~\cite{kamchatnov_condition_2011}.

\appendix

\section{Eigenvalue asymptotics}
\label{sec:numer-meth-cauchy}

We seek the dispersion relation $\Omega(k;\nu)$ of
Eq.~\eqref{eq:eigen-problem} for unstable transverse perturbations to
the shallow ($0 < \nu \ll 1$) dark line soliton \eqref{eq:17}.  Rather
than perform asymptotics directly on \eqref{eq:eigen-problem} it is
convenient to consider the eigenvalue problem in fluid
variables~\eqref{eq:Madelung}.  The soliton solution~\eqref{eq:17}
takes the form
\begin{subequations}
  \label{eq:62}
  \begin{align}
    \label{eq:3}
    \rho(x,y,t) &= \rho_s(\zeta) = 1 - \nu^2 \textrm{sech}^2(\zeta) ~, \\
    \label{eq:34}
    u(x,y,t) &= u_s(\zeta) = \frac{- \kappa}{\sinh^2 (\zeta) +
      \frac{\kappa^2}{\nu^2} \cosh^2 (\zeta)} ~, \\
    \label{eq:35}
    v(x,y,t) &= 0 ~, \quad \zeta = \nu(x-\kappa t) ~.
  \end{align}
\end{subequations}
Applying multiple scales to \eqref{eq:63} leads to the KP-I equation
for weakly nonlinear excitations of~\eqref{eq:NLS} to the uniform
state $\rho~\equiv 1$ (cf.~\cite{zakharov_multi-scale_1986}).  The
scalings involved motivate the following representation of weak
transverse perturbations to the dark soliton \eqref{eq:62}
\begin{align*}
    \rho(x,y,t) &= \rho_s(\zeta) - \eps f(\zeta) e^{i(p \eta - \Lambda
      \tau)} ~, \\
    u(x,y,t) &= u_s(\zeta) - \eps g(\zeta) e^{i(p \eta -
      \Lambda \tau)} ~,  \\
    v(x,y,t) &= \eps \nu
    h(\zeta) e^{i(p \eta - \Lambda \tau)} ~, \\
    \zeta &= \nu(x - \kappa t) ~, \quad \eta = \nu^2 y ~, \quad \tau = \nu^3
    t ~, \quad 0 < \eps \ll 1 ~.
\end{align*}
Inserting these expansions into \eqref{eq:63} and (\ref{eq:2}) while
keeping only $\mathcal{O}(\eps)$ terms gives
\begin{subequations}
  \label{eq:66}
  \begin{align}
    \label{eq:31}
    -f u_s'-g \rho _s'+f' \left[\kappa
      -u_s\right]+\rho _s
    \left(-g'+i p \nu ^2 h\right) &= -i \nu ^2 \Lambda f, \\[3mm]
    \nonumber
    -\frac{\nu ^2 \rho_s'}{\rho _s^3} \left(3 f'
      \rho_s'+4 f \rho_s''\right) + \frac{3 \nu ^2 f
      \rho_s'^3}{\rho_s^4}+4 g u_s'&  \\
    \nonumber
    +\frac{\nu ^2}{\rho
      _s^2} \left[\rho_s' \left(2 f''-p^2 \nu ^2 f\right)+2 f'
      \rho_s''(\zeta
      )+f \rho_s'''\right]&\\
    \label{eq:32}
    +4 \left[f'+g' \left(u_s-\kappa \right)\right]
    +\frac{1}{\rho_s} \left(p^2 \nu ^4 f'-\nu ^2 f'''\right) &= 4 i \nu
    ^2 \Lambda g ~, \\[3mm] 
    \label{eq:33}
    h' &= -ipg .
  \end{align}
\end{subequations}
This is an eigenvalue problem parametrized by $p$ and $\nu$ for the
eigenvalue $\Lambda = \Lambda(p;\nu)$ and eigenfunction $[f, g, h]^T$.

Assuming $p \notin \{0,\pm \pcut\}$ where $\pcut = \kcut/\nu^2$ so
that the eigenvalue of interest is simple, we expand\footnote{The two
  limits, linearization about the soliton and expanding in $\nu$, are
  not interchangeable.} the coefficient functions $\rho_s$ and $u_s$,
the parameter $\kappa = \sqrt{1 - \nu^2}$, the eigenfunction
$[f,g,h]^T$ and the eigenvalue $\Lambda$ in powers of $\nu^2$:
\begin{align}
  \nonumber
  f &= f_0 + \nu^2 f_1 + \nu^4 f_2 + \cdots ~, \\ 
  \label{eq:65}
  g &= f_0 + \nu^2 g_1 + \nu^4 g_2 + \cdots ~, \\ \nonumber
  h &= - i p \int \left ( f_0 + \nu^2 g_1 + \nu^4 g_2 + \cdots
  \right ) d\zeta ~, \\ \nonumber
  \Lambda(p;\nu) &= \Lambda_0(p) + \nu^2 \Lambda_1(p) + \cdots ~.
\end{align}
Then,~\eqref{eq:33} is automatically satisfied to all orders so we
only consider eqs.~\eqref{eq:31} and \eqref{eq:32} which expand,
respectively as
\begin{subequations}
  \label{eq:67}
  \begin{align}
    \label{eq:39}
    \Bigg \{ f_1'- g_1'+\left[ \left(2
        \text{sech}^2(\zeta)-\frac{1}{2}\right ) f_0 \right]'+i
    \Lambda_0 f_0 +p^2 \int f_0 d\zeta &\Bigg \} \\
    \nonumber + \nu^2 \Bigg \{ f_2' - g_2' -\frac{1}{2}
    \text{sech}^2(\zeta) \left(-2 \left(f_1'+g_1'\right)+f_0'+2 p^2 \int
      f_0 d\zeta \right
    )&\\
    \nonumber -\frac{1}{8} f_0'-\frac{1}{2} f_1' \text{sech}^4(\zeta)
    f_0'+\tanh(\zeta) \text{sech}^2(\zeta) \left(f_0-2
      \left(f_1+g_1\right)\right) &\\
    \nonumber +i \Lambda_1 f_0+i \Lambda_0 f_1 - 4 f_0 \tanh(\zeta)
    \text{sech}^4(\zeta) +p^2 \int g_1 d\zeta &\Bigg \}= {\cal
      O}(\nu^4) ~,
    \\
    \label{eq:40}
    \Bigg \{ g_1' -f_1' + \left[ \left (\text{sech}^2(\zeta )-\frac{1}{2}
      \right ) f_0 \right]'+\frac{1}{4}
    f_0''' + i \Lambda_0 f_0 & \Bigg \} \\
    \nonumber + \frac{\nu^2}{8} \Bigg \{ 8 g_2' - 8 f_2' +8 f_0
    \text{sech}^2(\zeta
    ) \tanh (\zeta ) \left(2 \text{sech}^2(\zeta )-1\right) &\\
    \nonumber + 2 \text{sech}^2(\zeta ) \left[-4 \tanh (\zeta )
      \left(f_0''+2
        g_1\right)+f_0'''+4 g_1' \right ] &\\
    \nonumber +2 \text{sech}^2(\zeta) \left(6-8 \text{sech}^2(\zeta )\right)
    f_0'-\left(2 p^2+1\right) f_0' &\\
    \nonumber +2 f_1'''-4 g_1'+8 i \Lambda _1 f_0+8 i \Lambda _0 g_1 &
    \Bigg \} = {\cal O}(\nu^4) ~.
  \end{align}
\end{subequations}

\subsection{KP eigenvalue problem}
\label{sec:kp-eigenv-probl}

Adding~\eqref{eq:39} to~\eqref{eq:40} gives
\begin{align*}
  \frac{1}{4} \left [ f_0''' - 4 \left [(1 - 3 \text{sech}^2(\zeta)) f_0
    \right ]' + 8 i \Lambda_0 f_0 + 4 p^2 \int f_0 d\zeta \right ] =
  {\cal O}(\nu^2) ~. 
\end{align*}
Differentiating and keeping only leading order terms gives
\begin{align*}
  \mathcal{L} f_0 &\doteq f_0'''' - 4 [(1 - 3 \textrm{sech}^2 \zeta )
  f_0]'' + 8 i \Lambda_0 f_0' + 4 p^2 f_0 = 0 ~.
\end{align*}
This is the KP eigenvalue problem studied in
\cite{alexander_transverse_1997}.  The unstable portion of the
spectrum includes one eigenpair
\begin{align*}
  f_0(\zeta;p) &= \frac{d^2}{d \zeta^2} \left \{ e^{\left (1 + \sqrt{1
          - 2p/\sqrt{3}} \right )\zeta} \left [ 2 - 2p/\sqrt{3} + 2
      \sqrt{1 - 2p/\sqrt{3}} \right ] \left [ 1 - \tanh(\zeta) \right
    ] 
  \right \} ~, \\
  \Lambda_0(p) &= -i\frac{p}{3} \sqrt{3 - 2 \sqrt{3} p}~, \quad 0 < p <
  \frac{\sqrt{3}}{2} \sim \pcut~, \quad 0 < \nu \ll 1 ~.
\end{align*}
This eigenvalue is continued onto the positive real line by the
eigenpair
\begin{align*}
  f_0(\zeta;p) &= \frac{d^2}{d \zeta^2} \left \{ e^{\left (1 - i
        \sqrt{ 2p/\sqrt{3}-1} \right )\zeta} \left [ 2 - 2p/\sqrt{3} +
      2 i\sqrt{2p/\sqrt{3}-1} \right ] \left [ 1 - \tanh(\zeta) \right
    ] 
  \right \} ~, \\
  \Lambda_0(p) &= \frac{p}{3} \sqrt{2\sqrt{3}p - 3}~, \quad p >
  \frac{\sqrt{3}}{2} \sim \pcut~, \quad 0 < \nu \ll 1 ~.
\end{align*}

\subsection{Perturbed KP eigenvalue problem}
\label{sec:pert-kp-eigenv}

Below we determine the correction $\Lambda_1(p)$.  $f_1$ is determined
in terms of $g_1$ by subtracting \eqref{eq:39} from~\eqref{eq:40} to
obtain
\begin{equation*}
  2 f_1' - 2 g_1' - \frac{1}{4} f_0''' + \left [\text{sech}^2(\zeta) f_0
  \right ]' + p^2
  \int f_0 d\zeta
  = \mathcal{O}(\nu^2)~,
\end{equation*}
so that
\begin{align}
  \label{eq:29}
    f_1 &= g_1 + \frac{1}{8} f_0'' - \frac{1}{2} \text{sech}^2(\zeta)
    f_0 - \frac{p^2}{2}
    \int \int f_0 d\zeta + \mathcal{O}(\nu^2) \\ \nonumber
    &\doteq g_1 + \ft + \mathcal{O}(\nu^2)~, \quad \ft = \frac{1}{8}
    f_0'' - \frac{1}{2} \text{sech}^2(\zeta) f_0 - \frac{p^2}{2}
    \int \int  f_0 d\zeta~.
\end{align}
Using \eqref{eq:29} in eqs.~\eqref{eq:39} and \eqref{eq:40}, adding
the two equations together and differentiating, the
$\mathcal{O}(\nu^2)$ terms equate to
\begin{align*}
  \nonumber
  \mathcal{L} g_1 =&~ -\text{sech}^2 (\zeta ) \left[4 \left(3
      f_0''-p^2 f_0+\tilde{f}''+4 \tilde{f}\right)+f_0''''\right] \\ \nonumber
  &~+2 \tanh (\zeta ) \text{sech}^2(\zeta ) \left(4 f_0'+3 f_0'''-4 p^2
    \int f_0 d\zeta+8 \tilde{f}'\right) \\
  &~ +8 \text{sech}^4(\zeta ) \left(2
    f_0''-4 f_0+3 \tilde{f}\right)+p^2f_0'' -4 i \Lambda_0
  \tilde{f}' \\
  \nonumber
  &~ -8 \tanh (\zeta ) \text{sech}^4(\zeta ) f_0'+f_0''(\zeta
  )-\tilde{f}''''+40 f_0 \text{sech}^6(\zeta
  )\\ \nonumber
  &~+2 \tilde{f}'' -8 i \Lambda_1 f_0' \\ \nonumber
  ~\doteq &~ G(\zeta;p) -8 i \Lambda_1 f_0' ~.
\end{align*}
Solvability then determines $\Lambda_1$
\begin{equation}
  \label{eq:41}
  \Lambda_1(p) = - i \frac{\int_{-\infty}^\infty G(\zeta;p) h^*(\zeta;p)
    d\zeta}{8 \int_{-\infty}^\infty f_0'(\zeta;p) h^*(\zeta;p) d\zeta} ~,
\end{equation}
where $h(\zeta;p)$ is the homogeneous solution of the adjoint problem
\begin{equation*}
  \mathcal{L}^\dagger h = h'''' - 4 (1 - 3 \textrm{sech}^2 \zeta )
  h'' + 8 i \Lambda_0^* 
  h' + 4 p^2 h = 0 ~.
\end{equation*}
Since $\Lambda_0$ is either purely real or purely imaginary, the
solution of the adjoint problem is 
\begin{equation*}
  \label{eq:48}
  h(\zeta;p) = \left \{
    \begin{array}{cc}
      f_0^*(\zeta;p)~, \quad p < \frac{\sqrt{3}}{2}~, \\
      f_0(\zeta;p)~, \quad p > \frac{\sqrt{3}}{2} 
    \end{array} \right . ~.
\end{equation*}

The integrals in~\eqref{eq:41} can be calculated explicitly
\begin{equation*}
  \Lambda_1(p) = \frac{p^2 \left (\sqrt{3} - p \right
    )}{6\sqrt{2 \sqrt{3} p - 3}} ~.
\end{equation*}
For the asymptotic expansion in \eqref{eq:65} to be valid, we require
$\Lambda_0(p) \gg \nu^2 \Lambda_1(p)$.  This puts a restriction on the
values of $p$ where the expansion is valid:
\begin{equation*}
  \pcut  < p \ll
  \frac{1}{\nu^2}  \quad \mathrm{or} \quad  0 < p < \pcut ~, \quad 0 <
  \nu \ll 1~.
\end{equation*}

Then, the unscaled eigenvalue $\Omega(k;\nu)$ in~\eqref{eq:eigen-problem}
has the asymptotic expansion
\begin{align*}
  \Omega(k;\nu) &\sim \nu^3 \Lambda_0(k/\nu^2) + \nu^5
  \Lambda_1(k/\nu^2)~, \\
  k &< \kcut(\nu)~, \quad \kcut(\nu) < k = \mathcal{O}(\nu^2) \ll 1~,
  \quad 0 < \nu \ll 1 \, ~,
\end{align*}
which is given in~\eqref{eq:asympt-eig}.

%%%%%%%%%%%%%%%%%%%%%%%%%%%%%%%%%%%%%%%%%%%%%%%%%%%%%%%%%

\renewcommand{\thetheorem}{B\arabic{theorem}}

\section{Theorem~\ref{theo:spectrum}}
\label{ap:spectrum}

In \cite{rousset_simple_2010} it was proved that $L_0$ has exactly one
negative eigenvalue which was determined explicitly in
\cite{kuznetsov-turitsyn-81} $(L_0 + \kcut^2/2) \mathbf{f} = 0$.  In
addition, it was proved in \cite{pego_eigenvalues_1992} from general
considerations of linear operators of the form $J L$ where $J$ is
skew-symmetric and $L$ is symmetric, that the number of eigenvalues of
$JL$ with a positive real part is at most the number of negative
eigenvalues of $L$.  The latter decomposition applies to~\eqref{eq:L},
where $L = L_0 + k^2/2$ is symmetric for $k \in \R$ and $J$ is
skew-symmetric.

Combining these results, for $0 < |k| < \kcut$, $k \in \R$, $L$ has
one negative eigenvalue and therefore $JL$ has at most one eigenvalue
with a positive real part.  By the instability of the dark soliton,
proven in \cite{rousset_simple_2010}, $JL$ has exactly one eigenvalue
with positive real part.  There is also exactly one eigenvalue with
negative real part via the following
\begin{lemma}
  \label{lem:imag}
  For $k \in \R$, the eigenvalues of $JL$ come in pairs of opposite
  sign.
\end{lemma}
\begin{proof}
  For any $k \in \R$,
  \begin{equation}
    \label{eq:16}
    JL \ = \  -R \sigma_3 LJ R \sigma_3 ~.
  \end{equation}

  Let $JL \mathbf{f} = \Gamma \mathbf{f}$.  Using~\eqref{eq:16} and
  one of the Pauli matrices~\eqref{eq:Pauli} gives
  $$
  JL R \sigma_3 \mathbf{f} = - \Gamma R \sigma_3 \mathbf{f}~.
  $$
  It follows that $(-\Gamma, R \sigma_3 \mathbf{f})$ is also an
  eigenpair for $JL$.
\end{proof}

On the other hand, for $|k| > \kcut$, $k \in \R$, $L$ has no negative
eigenvalues and therefore, by Lemma \ref{lem:imag}, $JL$ has only
purely imaginary eigenvalues.  We find numerically and asymptotically
in the shallow regime only two discrete, simple eigenvalues for $k \in
\C \setminus \{0, \pm \kcut \}$.

\section{Proof of Lemma~\ref{lem:adjoint}}
\label{ap:adjoint}

We make use of the following identity
\begin{align}
  \label{eq:14}
  J L_0 & =\ R \sigma_1 L_0 J R \sigma_1 ~.
\end{align}
For any $k \in \C$, consider $(\Omega_0,\fzero)$ an eigenpair for
eq.~(\ref{eq:eigen-problem}) satisfying
\begin{equation}
  \label{eq:37}
  \left [ J \left ( L_0 + \frac{1}{2}k^2 \right)
  + i \Omega_0 \right ] \fzero =  0 ~.
\end{equation}
Since $\Omega_0$ is a simple eigenvalue, it follows that
$$
\dim \left({\rm ker} \left \{ \left [JL_0 + \frac{1}{2}k^2 J + i\Omega_0
  \right ]^{\dagger} \right \}\right) = 1~.
$$
Therefore, it remains to verify that $\fzeroad$ is in the nullspace of
$[JL_0 + \frac{1}{2}k^2 J +i\Omega_0]^{\dagger}$.  We take the complex conjugate of
eq.~(\ref{eq:37}) and apply the decomposition~(\ref{eq:14}) to obtain
\begin{equation*}
  \label{eq:44}
  \left [ -R \sigma_1  L_0 J R \sigma_1 +
    \frac{1}{2}k^{*^2}J - i \Omega_0^* \right ] \fzero^* =  0 ~.
\end{equation*}
Applying $R\sigma_1$ yields
\begin{equation*}
  \label{eq:45}
  - \left [ L_0 J  + \frac{1}{2}k^{*^2} J  + i \Omega_0^*
  \right ] \fzeroad =  0 ~,
\end{equation*}
which is precisely the adjoint equation to (\ref{eq:37}).  Therefore,
we have
\begin{equation*}
  \label{eq:47}
  \ker \left \{ \left  [J \left (L_0 + \frac{1}{2} k^2 \right ) + i \Omega_0
    \right ]^\dagger 
  \right \} = \mathrm{span} \{ \fzeroad \} , \quad k \in
  \C .
\end{equation*}
By similar arguments with $J L_0 = -\sigma_2 L_0 J \sigma_2$, one can
show that $\sigma_2 \fzero \propto R\sigma_1 \fzero^*$ and hence spans
the kernel of $[J L_0 + \frac{1}{2}k^2 J + i\Omega_0 ]^\dagger$ when
$k, ~\Omega_0 \in \R$.  We use this null eigenfunction in our
numerical computations whenever $k \in (\kcut, \infty)$.

%\bibliography{BIB}

\begin{thebibliography}{10}
\expandafter\ifx\csname url\endcsname\relax
  \def\url#1{\texttt{#1}}\fi
\expandafter\ifx\csname urlprefix\endcsname\relax\def\urlprefix{URL }\fi
\expandafter\ifx\csname href\endcsname\relax
  \def\href#1#2{#2} \def\path#1{#1}\fi

\bibitem{kivshar_self-focusing_2000}
Y.~S. Kivshar, D.~E. Pelinovsky, Self-focusing and transverse instabilities of
  solitary waves, Phys. Rep. 331~(4) (2000) 117--195.

\bibitem{frantzeskakis_dark_2010}
D.~J. Frantzeskakis, Dark solitons in atomic {Bose-Einstein} condensates: from
  theory to experiments, J. Phys. A 43~(21) (2010) 213001.

\bibitem{bridges_transverse_2001}
T.~J. Bridges, Transverse instability of solitary-wave states of the water-wave
  problem, J. Fluid Mech. 439 (2001) 255--278.

\bibitem{akers_dynamics_2010}
B.~Akers, P.~A. Milewski, Dynamics of three-dimensional gravity-capillary
  solitary waves in deep water, {SIAM} J. Appl. Math. 70~(7) (2010) 2390--2408.

\bibitem{kadomtsev-petviashvili_1970} B.~B. Kadomtsev,
  V.~I. Petviashvili, On the stability of solitary waves in weakly
  dispersing media, Sov. Phys. Dokl. 15 (1970) 539--541.

\bibitem{zakharov__1975}
V.~E. Zakharov, Instability and nonlinear oscillations of solitons, {JETP}
  Lett. 22 (1975) 172--173.

\bibitem{zakharov-rubenchik-74}
V.~E. Zakharov, A.~M. Rubenchik, Instability of waveguides and solitons in
  nonlinear media, Sov. Phys. {JETP} 38 (1974) 494--500.

\bibitem{kuznetsov-turitsyn-81}
E.~A. Kuznetsov, K.~Turitsyn, Instability and collapse of solitons in media
  with a defocusing nonlinearity, Sov. Phys. {JETP} 67 (1981) 1583--1586.

\bibitem{el_oblique_2006}
G.~A. El, A.~Gammal, A.~M. Kamchatnov, Oblique dark solitons in supersonic flow
  of a {Bose-Einstein} condensate, Phys. Rev. Lett. 97~(18) (2006) 180405.

\bibitem{gladush_wave_2009}
Yu.~G. Gladush, A.~M. Kamchatnov, Z. Shi, P.~G. Kevrekidis,
D.~J. Frantzeskakis, B.~A. Malomed, Wave patterns generated by a
supersonic moving body in a binary Bose-Einstein condensate,
Phys. Rev. A 79 (2009) 033623.

\bibitem{kamchatnov-pitaevskii-08}
A.~M. Kamchatnov, L.~P. Pitaevski{\u i}, Stabilization of solitons generated by
  a supersonic flow of {Bose-Einstein} condensate past an obstacle, Phys. Rev.
  Lett. 100 (2008) 160402.

\bibitem{sturrock-58}
P.~A. Sturrock, Kinematics of growing waves, Phys. Rev. 112~(5) (1958)
  1488--1503.

\bibitem{briggs-64}
R.~G. Briggs, Electron-Stream Interaction with Plasmas, MIT Press, Cambridge,
  Mass., 1964.

\bibitem{dafermos_hyperbolic_2009}
C.~M. Dafermos, Hyperbolic Conservation Laws in Continuum Physics, 3rd Edition,
  Springer, 2009.

\bibitem{hoefer_dispersive_2009}
M.~A. Hoefer, M.~J. Ablowitz, Dispersive shock waves, Scholarpedia 4~(11)
  (2009) 5562.

\bibitem{dutton_observation_2001}
Z.~Dutton, M.~Budde, C.~Slowe, L.~V. Hau, Observation of quantum shock waves
  created with ultra-compressed slow light pulses in a {Bose-Einstein}
  condensate, Science 293 (2001) 663.

\bibitem{hoefer_dispersive_2006}
M.~A. Hoefer, M.~J. Ablowitz, I.~Coddington, E.~A. Cornell, P.~Engels,
  V.~Schweikhard, Dispersive and classical shock waves in {Bose-Einstein}
  condensates and gas dynamics, Phys. Rev. A 74 (2006) 023623.

\bibitem{meppelink_observation_2009}
R.~Meppelink, S.~B. Koller, J.~M. Vogels, P.~van~der Straten, E.~D. van Ooijen,
  N.~R. Heckenberg, H.~{Rubinsztein-Dunlop}, S.~A. Haine, M.~J. Davis,
  Observation of shock waves in a large {Bose-Einstein} condensate, Phys. Rev.
  A 80~(4) (2009) 043606.

\bibitem{wan_dispersive_2007}
W.~Wan, S.~Jia, J.~W. Fleischer, Dispersive superfluid-like shock waves in
  nonlinear optics, Nat. Phys. 3~(1) (2007) 46--51.

\bibitem{conti_observation_2009}
C.~Conti, A.~Fratalocchi, M.~Peccianti, G.~Ruocco, S.~Trillo, Observation of a
  gradient catastrophe generating solitons, Phys. Rev. Lett. 102~(8) (2009)
  083902.

\bibitem{armaroli_suppression_2009}
A.~Armaroli, S.~Trillo, A.~Fratalocchi, Suppression of transverse instabilities
  of dark solitons and their dispersive shock waves, Phys. Rev. A 80~(5).

\bibitem{hoefer_oblique_2009}
M.~A. Hoefer, B.~Ilan, Theory of two-dimensional oblique dispersive shock waves
  in supersonic flow of a superfluid, PRA 80 (2009) 061601(R).

\bibitem{majda_stability_1983}
A.~Majda, The stability of multidimensional shock fronts, Mem. Amer. Math. Soc.
  275.

\bibitem{li_analysis_2004}
D.~Li, Analysis on linear stability of oblique shock waves in steady supersonic
  flow, J. Diff. Eq. 207~(1) (2004) 195--225.

\bibitem{tkachev_courant-friedrichs_2009}
D.~L. Tkachev, A.~M. Blokhin, {Courant-Friedrich's} hypothesis and stability of
  the weak shock, Proc. Symp. Appl. Math. 67~(2) (2009) 959.

\bibitem{zumbrun_stability_2005}
K.~Zumbrun, H.~K. Jenssen, G.~Lyng, Stability of large-amplitude shock waves of
  compressible {Navier-Stokes} equations, in: S.~Friedlander, D.~Serre (Eds.),
  Handbook of Mathematical Fluid Dynamics, Vol.~3, Elsevier, Amsterdam, 2005,
  pp. 311--533.

\bibitem{courant_supersonic_1948}
R.~Courant, K.~O. Friedrichs, Supersonic Flow and Shock Waves,
  {Springer-Verlag}, Berlin, 1948.

\bibitem{alexander_transverse_1997}
J.~C. Alexander, R.~L. Pego, R.~L. Sachs, On the transverse instability of
  solitary waves in the {Kadomtsev-Petviashvili} equation, Phys. Lett. A
  226~(3-4) (1997) 187--192.

\bibitem{pelinovsky_self-focusing_1995}
D.~Pelinovsky, Y.~Stepanyants, Y.~Kivshar, Self-focusing of plane dark solitons
  in nonlinear defocusing media, Phys. Rev. E 51~(5) (1995) 5016--5026.

\bibitem{rousset_simple_2010}
F.~Rousset, N.~Tzvetkov, A simple criterion of transverse linear instability
  for solitary waves, Math. Res. Lett. 17~(1) (2010) 157--169.

\bibitem{pitaevskii_lifshitz_kinetics_1981}
L.~P. Pitaevski{\u i}, E.~Lifshitz, Physical Kinetics, Vol.~10 of Course of
  Theoretical Physics, Pergamon Press, Oxford, England, 1981, Ch.~62, pp.
  268--273.

\bibitem{infeld-rowlands-1990}
E.~Infeld, G.~Rowlands, Nonlinear Waves, Solitons and Chaos, Cambridge
  University Press, Cambridge, UK, 1990.

\bibitem{schmid_henningson_2001}
P.~J. Schmid, D.~S. Henningson, Stability and Transition in Shear Flows, Vol.
  142 of Applied Mathematical Sciences, Springer-Verlag, New York, 2000.

\bibitem{bers-83}
A.~Bers, Space-time evolution of plasma instabilities--absolute and convective,
  Vol.~1, Elsevier, Amsterdam: North-Holland, 1983, pp. 451--517.

\bibitem{huerre-monkewitz-85}
P.~Huerre, P.~A. Monkewitz, Absolute and convective instabilities in free shear
  layers, Journal of Fluid Mechanics 159 (1985) 151--168.

\bibitem{brevdo-bridges-96}
L.~Brevdo, T.~J. Bridges, Absolute and convective instabilities of spatially
  periodic flows, Phil. Trans. R. Soc. Lond. A 354~(1710) (1996) 1027--1064.

\bibitem{miller_applied_2006}
P.~D. Miller, Applied Asymptotic Analysis, {AMS} Publications, Providence,
  2006.

\bibitem{mironov_structure_2010}
V.~A. Mironov, A.~I. Smirnov, L.~A. Smirnov, Structure of vortex shedding past
  potential barriers moving in a {Bose-Einstein} condensate, JETP 110~(5)
  (2010) 877--889.

\bibitem{el_two-dimensional_2009}
G.~A. El, A.~M. Kamchatnov, V.~V. Khodorovskii, E.~S. Annibale, A.~Gammal,
  Two-dimensional supersonic nonlinear {S}chr\"{o}dinger equation flow past an
  extended obstacle, Phys. Rev. E 80 (2009) 046317.

\bibitem{hoefer_theory_2009}
M.~A. Hoefer, B.~Ilan, Theory of two-dimensional oblique dispersive shock waves
  in supersonic flow of a superfluid, Phys. Rev. A 80~(6) (2009) 061601(R).

\bibitem{whitham_non-linear_1965}
G.~B. Whitham, Non-linear dispersive waves, Proc. Roy. Soc. Ser. A 283 (1965)
  238--261.

\bibitem{gurevich_nonstationary_1974}
A.~V. Gurevich, L.~P. Pitaevski{\u i}, Nonstationary structure of a
  collisionless shock wave, Sov. Phys. {JETP} 38~(2) (1974) 291--297.

\bibitem{gurevich_dissipationless_1987}
A.~V. Gurevich, A.~L. Krylov, Dissipationless shock waves in media with
  positive dispersion, Sov. Phys. {JETP} 65~(5) (1987) 944--953.

\bibitem{el_unsteady_2006}
G.~A. El, R.~H.~J. Grimshaw, N.~F. Smyth, Unsteady undular bores in fully
  nonlinear shallow-water theory, Phys. Fluids 18~(2) (2006) 027104.

\bibitem{el_theory_2007}
G.~A. El, A.~Gammal, E.~G. Khamis, R.~A. Kraenkel, A.~M. Kamchatnov, Theory of
  optical dispersive shock waves in photorefractive media, Phys. Rev. A 76~(5)
  (2007) 053813.

\bibitem{el_resolution_2005}
G.~A. El, Resolution of a shock in hyperbolic systems modified by weak
  dispersion, Chaos 15 (2005) 037103.

\bibitem{gurevich_supersonic_1995}
A.~V. Gurevich, A.~L. Krylov, V.~V. Khodorovskii, G.~A. El, Supersonic flow
  past bodies in dispersive hydrodynamics, Sov. Phys. {JETP} 81 (1995) 87--96.

\bibitem{gurevich_supersonic_1996}
A.~V. Gurevich, A.~L. Krylov, V.~V. Khodorovskii, G.~A. El, Supersonic flow
  past finite-length bodies in dispersive hydrodynamics, Sov. Phys. {JETP} 82
  (1996) 709--718.

\bibitem{el_spatial_2006}
G.~A. El, A.~M. Kamchatnov, Spatial dispersive shock waves generated in
  supersonic flow of {Bose-Einstein} condensate past slender body, Phys. Rev. A
  350~(3-4) (2006) 192--196.

\bibitem{kamchatnov_condition_2011}
A.~M. Kamchatnov, S.~V. Korneev, Condition for convective instability of dark
  solitons, arXiv:1105.0789 [cond-mat.quant-gas].

\bibitem{hayes_hypersonic_2004}
W.~D. Hayes, R.~F. Probstein, Hypersonic Inviscid Flow, Dover, 2004.

\bibitem{hoefer_piston_2008}
M.~A. Hoefer, M.~J. Ablowitz, P.~Engels, Piston dispersive shock wave problem,
  Phys. Rev. Lett. 100 (2008) 084504.

\bibitem{kamchatnov_flow_2010}
A.~Kamchatnov, S.~Korneev, Flow of a {Bose-Einstein} condensate in a
  quasi-one-dimensional channel under the action of a piston, JETP 110~(1)
  (2010) 170--182.

\bibitem{keetels_fourier_2007}
G.~Keetels, U.~{D'Ortona}, W.~Kramer, H.~Clercx, K.~Schneider, G.~van Heijst,
  Fourier spectral and wavelet solvers for the incompressible {Navier-Stokes}
  equations with volume-penalization: Convergence of a dipole-wall collision,
  J. Comp. Phys. 227~(2) (2007) 919--945.

\bibitem{ablowitz_spectral_2005}
M.~J. Ablowitz, Z.~H. Musslimani, Spectral renormalization method for computing
  self-localized solutions to nonlinear systems, Opt. Lett. 30 (2005)
  2140--2142.

\bibitem{gray_weak_2007}
J.~M. N.~T. Gray, X.~Cui, Weak, strong and detached oblique shocks in
  {Gravity-Driven} granular {Free-Surface} flows, J. Fluid Mech. 579 (2007)
  113--136.

\bibitem{pego_eigenvalues_1992}
R.~L. Pego, M.~I. Weinstein, Eigenvalues, and instabilities of solitary waves,
  Phil. Trans. Roy. Soc. London Ser. A 340~(1656) (1992) 47--94.

\bibitem{zakharov_multi-scale_1986}
V.~Zakharov, E.~Kuznetsov, Multi-scale expansions in the theory of systems
  integrable by the inverse scattering transform, Physica D 18~(1-3) (1986)
  455--463.

\end{thebibliography}
\bibliographystyle{elsarticle-num}

\end{document}